\documentclass[12pt,reqno,a4paper]{amsart}
\usepackage[left=3cm, top=3.5cm, bottom=3.5cm, right=3cm]{geometry}
\usepackage{mathrsfs}
\usepackage{amsmath,amsxtra,amsfonts,amssymb, amsthm, amscd, epsfig}
\usepackage[dvipsnames,usenames]{color}
\usepackage{amssymb,amsmath,amsfonts,epsfig}
\usepackage{graphicx}

\usepackage{tikz}
\usepackage{cite}

\numberwithin{equation}{section}

\newtheorem{thm}{Theorem}[section]
\newtheorem{prop}[thm]{Proposition}
\newtheorem{cor}[thm]{Corollary}

\newtheorem{lem}[thm]{Lemma}

\newcommand{\eqa}{\begin{eqnarray}}
\newcommand{\eeqa}{\end{eqnarray}}
\newcommand{\beq}{\begin{equation}}
\newcommand{\eeq}{\end{equation}}
\newcommand{\nn}{\nonumber}
\newcommand{\p}{\partial}
\def \la {\langle}
\def \ra{\rangle}

 \def\res{\mathop{\text{\rm Res}}}

\def \dsum{\displaystyle\sum}

\allowdisplaybreaks[1]

\newcommand\Z{\mathbb{Z}}

\newcommand\om{\omega}

\newcommand\dt{\delta}

\def \la {\langle} \def \ra{\rangle}

\begin{document}

\title[]
{Infinite-dimensional Frobenius Manifolds and extensions of Genus-Zero Whitham Hierarchies }

\author[]{Shilin Ma}

\address[]{Shilin Ma,  School of Mathematics and Statistics,  Xi'an Jiaotong University, Xi'an 710049,  P. R.
	China,}
\email{mashilin@xjtu.edu.cn}
\author[]{Chao-Zhong Wu}

\address[]{Chao-Zhong Wu,  School of Mathematics, Sun Yat-Sen University, Guangzhou 510275, P. R. China,}
\email{wuchaozhong@sysu.edu.cn}

\author[]{Dafeng Zuo}

\address[]{Dafeng Zuo,  School of Mathematical Sciences, University of Science and Technology of China, Hefei 230026, P. R. China,}
\email{dfzuo@ustc.edu.cn}

\date{\today}

\begin{abstract}
In this paper we construct a class of infinite-dimensional Frobenius
manifolds in the spaces of pairs of meromorphic functions defined on certain regions of the Riemann sphere. For such Frobenius manifolds, we obtain their principal hierarchies and show them to be extensions of the genus-zero Whitham hierarchies.

\vskip 2ex

\end{abstract}

\maketitle 


\section{Introduction}

\subsection{Background}	

The concept of Frobenius manifold, introduced by Dubrovin \cite{dub1998}, provides a geometric framework for the Witten-Dijkgraaf-Verlinde-Verlinde (WDVV) equations in two-dimensional topological field theory. Briefly, a Frobenius manifold structure consists of Frobenius algebras on the tangent spaces that depend smoothly on the base points, a flat unit vector field with respect to the Levi-Civita connection induced by a flat metric, a symmetric $3$-tensor, a locally existing potential, and a specific Euler vector field to characterize quasi-homogeneous properties. This notion has found significant applications in various branches of mathematical physics, such as Gromov-Witten theory, singularity theory, and integrable systems; see, for example, \cite{manin1999frobenius,hertling2002frobenius,dubrovin2001normal,givental2005simple} and references therein. In particular, a (finite-dimensional) Frobenius manifold is naturally associated with a Hamiltonian integrable hierarchy, referred to as the principal hierarchy \cite{dub1998}. This principal hierarchy, composed of a system of $1+1$ dispersionless equations, can be uniquely deformed into the Dubrovin-Zhang hierarchy under certain axioms \cite{dubrovin2001normal}.

In recent years, studying the relationship between integrable hierarchies of $2+1$ equations and infinite-dimensional Frobenius manifolds has become an intriguing topic. The first example of an infinite-dimensional Frobenius manifold was constructed by Carlet, Dubrovin, and Mertens \cite{carlet2011infinite} on a space of pairs of certain meromorphic functions defined on a neighborhood of the unit circle.
Following this approach, additional examples of infinite-dimensional Frobenius manifolds have been discovered, underlying integrable hierarchies of $2+1$ equations, such as the two-component Kadomtsev-Petviashvili (KP) hierarchy of Type B \cite{wu2012class}, the Toda lattice hierarchy \cite{wu2014infinite}, and an extended KP hierarchy \cite{ma2021infinite}. The principal hierarchies of these infinite-dimensional Frobenius manifolds were investigated in \cite{carlet2015principal,wu2012class}.
It is worth noting that infinite-dimensional Frobenius manifolds can also be constructed via alternative methods, such as using Schwartz functions \cite{raimondo2012frobenius} or classical $r$-matrices \cite{szablikowski2015classical}.

In this paper, we study a class of infinite-dimensional Frobenius manifolds and their principal hierarchies associated with the genus-zero Whitham hierarchies, following the framework introduced in \cite{carlet2011infinite,carlet2015principal}. Recall that Krichever \cite{krichever1988method,krichever1994tau} introduced a universal Whitham hierarchy on the moduli space of Riemann surfaces of arbitrary genus with $m+1$ marked points. He also defined the tau function for this hierarchy and explored its applications in topological gravity. The universal Whitham hierarchies have been applied to Seiberg-Witten theory \cite{gorsky1995integrability,itoyama1996integrability,itoyama1997prepotential,krichever1997integrable,gorsky1998rg}, Laplacian growth \cite{mineev2001whitham,krichever2004laplacian,zabrodin2005whitham,abanov2009multi}, and conformal mapping \cite{wiegmann2000conformal,martinez2006genus}. In the genus-zero case, consider the Riemann sphere with the infinity point and $m$ movable points $\varphi_i$ marked,  where $ i=1,2,\dots m$. This setup corresponds to the following formal Laurent series in $z$:

\begin{equation}\label{whihier}
	\lambda_{i}(z)=\left\{
\begin{array}{cl}
  z+\sum_{p \geq 1} v_{0,p}z^{-p},  & i=0; \\
 \sum_{p \geq-1} {v}_{i,p}(z-\varphi_{i})^p,\quad & i=1,\cdots,m,
\end{array}
\right.
\end{equation}
where $v_{i,p}$ and $\varphi_i$ are unknown functions. The genus-zero Whitham hierarchy is composed of a system of evolutionary equations as follows
\begin{align}
	& 
\frac{\partial \lambda_{i}(z)}{\partial s^{0,k}}=\{\left(\lambda_{0}(z)^k\right)_{\infty,\,\ge 0}, \lambda_{i}(z)\}, \label{whi1}\\
	& 
\frac{\partial \lambda_{i}(z)}{\partial s^{j,k}}=\{-(\lambda_{j}(z)^k)_{\varphi_{j},\,< 0}, \lambda_{i}(z)\},\label{whi2}
\\
& 
\frac{\partial \lambda_{i}(z)}{\partial s^{j,0}}=\{\log (z-\varphi_{j}), \lambda_{i}(z)\}, \label{whi3}
\end{align}
with $i\in\{0,1,\dots,m\}$, $j\in\{1,\cdots,m\}$ and $k\in\mathbb{Z}_{>0}$. Here we use the truncations of Laurent series:
\[
\left(\sum_{p} f_p z^p \right)_{\infty,\,\ge0}=\sum_{p\ge0} f_p z^p, \quad
\left(\sum_{p} f_p (z-\varphi_i)^p \right)_{\varphi_i,\,<0}=\sum_{p<0} f_p (z-\varphi_i)^p,
\]
and the Poisson bracket $\{\ ,\ \}$ means
\begin{equation}\label{bracket}
	\{f, g\}:=\frac{\partial f}{\partial z} \frac{\partial g}{\partial x}-\frac{\partial g}{\partial z} \frac{\partial f}{\partial x}.
\end{equation}

In particular, when  $m = 0$, this hierarchy reduces to the dispersionless KP hierarchy.  Whenever $m = 1$, we constructed a series of infinite-dimensional Frobenius manifolds in \cite{ma2021infinite}, where the principal hierarchies remain unexplored. Here, we generalize the construction in \cite{ma2021infinite} to arbitrary  $m \geq 1$ and further investigate the principal hierarchies. This generalization is highly nontrivial for several reasons.

Firstly, unlike previous works \cite{carlet2011infinite,wu2012class,wu2014infinite,ma2021infinite}, which considered spaces of pairs of meromorphic functions on a neighborhood of the unit circle, our current framework must account for meromorphic functions defined near $m$  distinct circles on the Riemann sphere. This introduces complexities in truncating these functions appropriately.

Secondly, while the bi-Hamiltonian structure of the Whitham hierarchy \eqref{whi1}-\eqref{whi3} is established for  $m = 1$, it remains unknown for  $m \geq 2$ (cf. \cite{wu2016extension}). To clarify the relationship between this hierarchy and our infinite-dimensional Frobenius manifolds, we must derive the principal hierarchies for the latter.

Thirdly, when formulating the principal hierarchy, we encounter challenges due to the nontrivial  $R$-matrix of the infinite-dimensional Frobenius manifold, particularly in resonant cases.

Our motivation also stems from the pursuit of a dispersive version of the genus-zero Whitham hierarchies. Szablikowski and Blaszak \cite{szablikowski2008dispersionful} made early progress in this direction using Lax equations for formal pseudo-differential operators. For  $m = 1$, a dispersive version was proposed in \cite{geng2024kp}, where tau functions linked to Baker-Akhiezer functions of the full hierarchy were defined. Inspired by Dubrovin and Zhang \cite{dubrovin2001normal}, we anticipate that deforming the principal hierarchy of the infinite-dimensional Frobenius manifolds constructed here may yield novel dispersive hierarchies.

\subsection{Main results}
\label{sec-M}

Let \( D_1, \ldots, D_m \) be pairwise non-intersecting closed disks in the Riemann sphere \( \mathbb{P}^{1}=\mathbb{C}\cup\{\infty\} \) ,none of which contain $\infty$. Denote by $\gamma_i=\partial D_i$ the  boundary of $D_i$, and let $\gamma=\bigcup_{i=1}^m \gamma_i$. The $m$-circle $\gamma$ divides the Riemann sphere into two open sets: the interior open set \( \mathbf{D}^{int}=\bigcup_{i=1}^m (D_i\setminus\gamma_i) \) and the exterior domain \( \mathbf{D}^{ext}=\mathbb{P}^1\setminus \bigcup_{i=1}^m D_i  \).

Let \( \mathcal{H} \) be the space of germs of functions that are holomorphic in some neighborhood of \( \gamma \). For any \( f(z) \in \mathcal{H} \), its positive part and negative part are defined as follows:
\begin{align*}
&f(z)_+ := \frac{1}{2\pi \mathrm{i}} \sum_{s=1}^{m} \oint_{\gamma_s} \frac{f(p)}{p-z} \, dp, \quad z \in {\mathbf{D}}^{int};  \\
&f(z)_- :=\left\{
\begin{array}{cl} \displaystyle
   -\frac{1}{2\pi \mathrm{i}} \sum_{s=1}^{m} \oint_{\gamma_s} \frac{f(p)}{p-z} \, dp, & \quad z \in \mathbf{D}^{ext}\setminus\{\infty\}; \\
  0, & z=\infty.
\end{array}
 \right.
\end{align*}
Both \( f(z)_+ \) and \( f(z)_- \) can be analytically extended beyond the boundary $\gamma$, such that in a neighborhood of $\gamma$ one has
\[ f(z) = f(z)_+ + f(z)_- .
\]
For example, the set $\mathcal{H}$ contains the characteristic functions \( \mathbf{1}_{\gamma_i}(z) \) given by
\begin{equation}\label{1fun}
\mathbf{1}_{\gamma_i}(z)|_{\gamma_j} = \delta_{ij}, \quad  i, j \in\{ 1, \ldots, m\}.
\end{equation}
It is a subtle fact that, given $f(z)\in\mathcal{H}$, although  $f(z)\mathbf{1}_{\gamma_{j}}(z)$ is zero on
$\gamma_{j'}$ for $j'\ne j$, the projections $(f(z)\mathbf{1}_{\gamma_{j}}(z))_{\pm}$ may be not.



 Let \( n_0, \ldots, n_m \) be positive integers, and  \( d_1, \ldots, d_m \)  be non-zero integers. Define $\mathcal{M}=\mathcal{M}_{n_{0},n_{1},\cdots,n_{m}}^{d_{1},d_{2},\cdots,d_{m}}$ to be a subset of \( \mathcal{H} \times \mathcal{H} \), consisting of pairs \( (a(z), \hat{a}(z)) \), where \( a(z) \) extends to a meromorphic function on \( \mathbf{D}^{ext}  \) with a single pole \( \{ \infty \} \), and \( \hat{a}(z) \) extends to a meromorphic function on \( \mathbf{D}^{int} \) with poles \( \varphi_i\in (D_i\setminus\gamma_i)  \) for \( i = 1, \ldots, m \).
 More exactly, suppose that  \( a(z) \) and \( \hat{a}(z) \) have expansions:
\begin{align*}
a(z)=z^{n_{0}}+a_{n_{0}-2}z^{n_{0}-2}+a_{n_{0}-3}z^{n_{0}-3}+\cdots,&\quad z\to \infty; \\	
\hat{a}(z)=\hat{a}_{j,-n_{j}}(z-\varphi_{j})^{-n_{j}}+\hat{a}_{j,-n_{j}+1} (z-\varphi_{j})^{-n_{j}+1}+\cdots,& \quad z\to \varphi_{j},
\end{align*}
with $j\in\{1,2,\cdots,m\}$,
and that the following conditions are satisfied:
\begin{enumerate}
	\item[(C1)] $\hat{a}_{j,-n_{j}}\ne 0$ for all $j\in\{1,2,\cdots,m\}$;
	\item[(C2)] Letting
\begin{equation}\label{zeta0}
 \zeta(z) = a(z) - \hat{a}(z),
\end{equation}
for each \( j\in\{1,2,\cdots,m\} \) there exists a holomorphic function \( w_j(z) \) defined in a neighborhood of \( \gamma_j \) such that
	\[
	\zeta(z)|_{\gamma_j} = w_j(z)^{d_j}|_{\gamma_j}, \quad w'(z)|_{\gamma_j} \neq 0,
	\]
	and \( w_j(z) \) maps \( \gamma_j \) to a path around the origin with the winding number \( 1 \).
\end{enumerate}
Such a subset can be viewed as an infinite-dimensional manifold, with coordinates given by $\{a_{j}\}_{j\le n_{0}-2}\cup\{\hat{a}_{1,j}\}_{j\ge-n_{1}}\cup\cdots \cup \{\hat{a}_{m,j}\}_{j\ge-n_{m}}\cup \{\varphi_{j}\}_{j=1}^{m}$. We then have 

\begin{thm}\label{mainthm1}
On $\mathcal{M}$ there exists an infinite-dimensional Frobenius manifold structure \((\mathcal{M}, \circ, e, \eta, E)\) of charge $d=1-\frac{2}{n_{0}}$, where the flat metric \(\eta\), the invariant multiplication  \(\circ\) of vector fields with respect to $\eta$, the unity vector field \(e\) and the Euler vector field \(E\) are given by \eqref{whimetric}, \eqref{whimul}, \eqref{whiid} and \eqref{E} below, respectively.
\end{thm}

For \(\mathcal{M}\), we can choose its flat coordinates with respect to the metric \(\eta\) (see Section~\ref{whimetricsec} below):
\[
\mathbf{u}=\{t_{i,s} | 1 \leq i \leq m, \ s \in \mathbb{Z}\} \cup \{h_{0,j} | 1 \leq j \leq n_0 - 1\} \cup \{h_{k,r} | 1 \leq k \leq m, \ 0 \leq r \leq n_k\}.
\]

\begin{thm}\label{mainthm3}
	The principal hierarchy for the infinite-dimensional Frobenius manifold
	$\mathcal{M}$ takes the form: for $t_{i,s},h_{0,j},h_{k,r}\in \mathbf{u} $,
	\begin{align*}
		\frac{\partial}{\partial T^{t_{i,s},\, p}}(a(z),\hat{a}(z)) =&\frac{1}{p!}\frac{d_{i}}{s+d_{i}} \left(-\left\{(\zeta^{\frac{s}{d_{i}}}\left(a^{p}-\hat{a}^{p}\right)\mathbf{1}_{\gamma_{i}})_{-},a\right\}, \left\{(\zeta^{\frac{s}{d_{i}}}\left(a^{p}-\hat{a}^{p}\right)\mathbf{1}_{\gamma_{i}})_{+},\hat{a}\right\}\right);
\\
\frac{\partial}{\partial T^{h_{0,j},\, p} }(a(z),\hat{a}(z))  	= &	\frac{\Gamma\left(1-\frac{j}{n_{0}}\right)}{\Gamma\left(2+p-\frac{j} {n_{0}}\right)}\left(\left\{(a^{1+p-\frac{j}{n_{0}}})_{\infty,\,\ge 0},a\right\},\left\{(a^{1+p-\frac{j}{n_{0}}})_{\infty,\,\ge 0},\hat{a}\right\}\right);
\\
\frac{\partial}{\partial T^{h_{k,r},\, p}}(a(z),\hat{a}(z))
			=& 	\frac{\Gamma\left(1 - \frac{r}{n_{k}}\right)}{\Gamma\left(2 + p - \frac{r}{n_{k}}\right)}\left( \left\{-( \hat{a}^{1 + p - \frac{r}{n_{k}}})_{\varphi_{k}, \leq -1}, a\right\}, \left\{-( \hat{a}^{1 + p - \frac{r}{n_{k}}})_{\varphi_{k}, \leq -1}, \hat{a}\right\} \right)
		\end{align*}
with $s\ne -d_{i},r\ne n_{k}$,	and
		\begin{align*}
&\frac{\partial}{\partial T^{t_{i,-d_{i}},\, p}}(a(z),\hat{a}(z)) \\
=&{d_{i}}\bigg(-\left\{\left(\frac{a^{p}}{p!}(\log\frac{a^{\frac{1}{n_{0}}}}{z-\varphi_{i}} -\frac{c_{p}}{n_{0}})\right)_{\infty,\,\ge 0},a\right\}+\left\{\left(\frac{a^{p}}{p!}\log\frac{z-\varphi_{i}}{\zeta^{\frac{1}{d_{i}}}} \mathbf{1}_{\gamma_{i}}\right)_{-},a\right\}\\
			&+\left\{\sum_{j\ne i}\left(\frac{a^{p}}{p!}\log(z-\varphi_{i})\mathbf{1}_{\gamma_{j}}\right)_{-},a\right\} -\left\{\frac{a^{p}}{p!}\log(z-\varphi_{i}),a\right\},
\\		&-\left\{\left(\frac{a^{p}}{p!}(\log\frac{a^{\frac{1}{n_{0}}}}{z-\varphi_{i}}-\frac{c_{p}}{n_{0}})\right)_{\infty,\,\ge 0},\hat{a}\right\}-\left\{\left(\frac{a^{p}}{p!}\log\frac{z-\varphi_{i}}{\zeta^{\frac{1}{d_{i}}}}\mathbf{1}_{\gamma_{i}}\right)_{+},\hat{a}\right\}\\
			&-\left\{\sum_{j\ne i}\left(\frac{a^{p}}{p!}\log(z-\varphi_{i})\mathbf{1}_{\gamma_{j}}\right)_{+},\hat{a}\right\}\bigg),
\\
&\frac{\partial}{\partial T^{h_{k,n_{k}},\, p}}(a(z),\hat{a}(z)) \\
=&n_{k}\bigg(-\left\{\left(\frac{\hat{a}^{p}}{p!}(\log(z-\varphi_{k})\hat{a}^{\frac{1}{n_{k}}}-\frac{c_{p}}{n_{k}})\right)_{\varphi_{k},\,\le -1},a\right\}+\left\{\left(\frac{a^{p}}{p!}(\log\frac{a^{\frac{1}{n_{0}}}}{z-\varphi_{k}}-\frac{c_{p}}{n_{0}})\right)_{\infty,\,\ge 0},a\right\}\\
			&+\left\{\left(\frac{\hat{a}^{p}}{p!}\log\frac{z-\varphi_{k}}{\zeta^{\frac{1}{d_{k}}}}\mathbf{1}_{\gamma_{k}}\right)_{-},a\right\}-\left\{\left(\frac{a^{p}}{p!}\log\frac{z-\varphi_{k}}{\zeta^{\frac{1}{d_{k}}}}\mathbf{1}_{\gamma_{k}}\right)_{-},a\right\}\\
			&+\left\{\frac{a^{p}}{p!}\log(z-\varphi_{k}),a\right\}-\sum_{j\ne k}\left\{\left(\frac{a^{p}}{p!}\log(z-\varphi_{k})\mathbf{1}_{\gamma_{j}}\right)_{-},a\right\},\\
			&-\left\{\left(\frac{\hat{a}^{p}}{p!}(\log(z-\varphi_{k})\hat{a}^{\frac{1}{n_{k}}}-\frac{c_{p}}{n_{k}})\right)_{\varphi_{k},\,\le -1},\hat{a}\right\}+\left\{\left(\frac{a^{p}}{p!}(\log\frac{a^{\frac{1}{n_{0}}}}{z-\varphi_{k}}-\frac{c_{p}}{n_{0}})\right)_{\infty,\,\ge 0},\hat{a}\right\}\\
			&-\left\{\left(\frac{\hat{a}^{p}}{p!}\log\frac{z-\varphi_{k}}{\zeta^{\frac{1}{d_{k}}}}\mathbf{1}_{\gamma_{k}}\right)_{+},\hat{a}\right\}+\left\{\left(\frac{a^{p}}{p!}\log\frac{z-\varphi_{k}}{\zeta^{\frac{1}{d_{k}}}}\mathbf{1}_{\gamma_{k}}\right)_{+},\hat{a}\right\}\\
			&+\left\{\frac{\hat{a}^{p}}{p!}\log(z-\varphi_{k})\mathbf{1}_{\gamma_{k}},\hat{a}\right\}+\sum_{j\ne k}\left\{\left(\frac{a^{p}}{p!}\log(z-\varphi_{k})\mathbf{1}_{\gamma_{j}}\right)_{+},\hat{a}\right\}\bigg).
	\end{align*}
Here recall that $\mathbf{1}_{\gamma_{i}}$ and $\zeta(z)$ are given in \eqref{1fun} and \eqref{zeta0} respectively.
\end{thm}


\begin{cor}\label{prinwhi}
Suppose that the Laurent series \(\lambda_i(z)\) are given in \eqref{whihier}, and 
\begin{equation}\label{assum}
a(z) = \lambda_0(z)^{n_0} \hbox{ as } z\to\infty; \qquad \hat{a}(z)= \lambda_l(z)^{n_l} \hbox{ as } z\to \varphi_{l},\quad l = 1, \ldots, m.
\end{equation}
Then the genus-zero Whitham hierarchy \eqref{whi1}--\eqref{whi3} is a subhierarchy of the principal hierarchy for the infinite-dimensional Frobenius manifold $\mathcal{M}$.
\end{cor}

This paper is organized as follows. In Section 2, we construct the infinite-dimensional Frobenius manifold structure \((\mathcal{M}, \circ, e, \eta, E)\) and prove Theorem~\ref{mainthm1}. In Section 3 we  derive a Poisson pencil on the loop space of \(\mathcal{M}\), then prove Theorem \ref{mainthm3} and Corollary~\ref{prinwhi}; furthermore, we  illustrate how to derive finite-dimensional Frobenius manifolds with rational superpotentials  from $\mathcal{M}$.  Section 4 concludes this paper.

\section{Infinite-dimensional Frobenius manifold structure on $\mathcal{M}$}

In this section, we assign an infinite-dimensional Frobenius manifold structure on
 $\mathcal{M}$. We will firstly introduce a flat metric on $\mathcal{M}$ and choose a set of flat coordinates. Then we consider a map, which is induced by the flat metric, from the cotangent space to the tangent space at each point. With the help of this map, we define multiplication of tangent vector fields, find out the unit vector, write down the symmetric $3$-tensor, and then deduce the potential of the infinite-dimensional Frobenius manifold. Finally, we find the Euler vector field, and study the intersection form on $\mathcal{M}$.

\subsection{Flat metric}\label{whimetricsec}
Recall the space $\mathcal{M}$ consisting of pairs of functions $(a(z),\hat{a}(z))$ satisfying conditions (C1) and (C2) in the previous section. Let
\begin{equation}\label{zetaell}
	\zeta(z)=a(z)-\hat{a}(z),\quad \ell(z)=a(z)_{+}+\hat{a}(z)_{-}.
\end{equation}
Note that $\zeta(z)$ is defined in a neighborhood of $\gamma$, while $\ell(z)$ can be analytically continued to $\mathbb{C}\setminus\{\varphi_1,\dots,\varphi_m\}$. It is easy to see
\begin{equation}\label{zetaell2a}
	a(z)=\ell(z)+\zeta(z)_-, \quad \hat{a}(z)=\ell(z)-\zeta(z)_+.
\end{equation}
Hence the pairs $(a(z),\hat{a}(z))$ and $(\zeta(z), \ell(z))$ are determined by each other.

Given $\vec{a}=(a(z),\hat{a}(z))\in\mathcal{M}$, we identify each vector
$X$ in the tangent space $T_{\vec{a}}\mathcal{M}$ with the derivative  $(\p_{X}
a(z), \p_{X}\hat{a}(z))$ along $X$.
Let us introduce a bilinear form on $T_{\vec{a}}\mathcal{M}$ as
\begin{equation}\label{whimetric}
	\langle X_1,X_2\rangle_{\eta}=
	-\frac{1}{2\pi\mathrm{i}}\dsum_{j=1}^{m}\oint_{\gamma_{j}}\frac{\p_{1}\zeta(z) \cdot\p_{2}\zeta(z)}{\zeta'(z)}dz
	-\left(\res_{z=\infty}+\dsum_{j=1}^{m}\res_{z=\varphi_{j}}\right)\frac{\p_{1}\ell(z)\cdot \p_{2}\ell(z)}{\ell'(z)}dz,
\end{equation}
where
$\p_{\nu}=\p_{X_{\nu}}$ for $\nu=1,2$. Clearly this bilinear form is symmetric and nondegenerate, hence it can be regarded as a metric on $\mathcal{M}$.

We proceed to select a system of flat coordinates with respect to the metric \eqref{whimetric}, with a similar method as in \cite{ma2021infinite}.
According to condition (C2) in the definition of \(\mathcal{M}\), for each \( j \in\{ 1, 2, \ldots, m\} \) there exists a holomorphic function \( w_j(z) \) defined in a neighborhood of \( \gamma_j \) such that
\[
\zeta(z)|_{\gamma_j} = w_j(z)^{d_j}|_{\gamma_j}, \quad w'(z)|_{\gamma_j} \neq 0,
\]
and that \( w_j(z) \) maps \( \gamma_j \) to a path $\tilde{\gamma}_j$ encircling the origin with  winding number \( 1 \). The inverse function $z(w_{j})$ of $w_j(z)$ can be expanded as
\begin{equation}\label{texp}
	z(w_{j})=\sum_{s \in \mathbb{Z}} t_{j,s} w_{j}^s, \quad w_{j} \in \tilde{\gamma}_j,
\end{equation}
where the coefficients $t_{j,s}$ are given by
\[
t_{j,s}=\frac{1}{2 \pi \mathrm{i}} \oint_{\tilde{\gamma}_j} z(w_{j}) w_{j}^{-s-1} d w_{j}, \quad  j\in\{1,2,\cdots,m\},\ s \in \mathbb{Z}.
\]
There exists a family of functions \( \chi_j(z) \), defined on some punctured neighborhood around \( \infty \) for \( j = 0 \) and around \( \varphi_j \) for \( j \in\{ 1,2, \ldots, m\} \), such that
$$
\begin{aligned}
	& \chi_{0}(z)=\ell(z)^{\frac{1}{n_{0}}}=z+b_{0,1} z^{-1}+b_{0,2} z^{-2}+\cdots, \quad z \rightarrow \infty, \\
	& \chi_{j}(z)=\ell(z)^{\frac{1}{n_{j}}}=\hat{a}_{j,-n_{j}}^{\frac{1}{n}}(z-\varphi_{j})^{-1}+b_{j,0}+b_{j,1}(z-\varphi_{j})+\cdots, \quad z \rightarrow \varphi_{j} .
\end{aligned}
$$
The inverse functions  $z(\chi_{j})$ of $\chi_j(z)$ can be expanded into Laurent series of the form:
\begin{align}
	& z(\chi_{0})=\chi_{0}-h_{0,1} \chi_{0}^{-1}-h_{0,2} \chi_{0}^{-2}-\cdots-h_{0,n_{0}-1} \chi_{0}^{-n_{0}+1}+\cdots, \quad z \rightarrow \infty ,\label{hexp}\\
	& z(\chi_{j})=h_{j,0}+h_{j,1} \chi_{j}^{-1}+h_{j,2} \chi_{j}^{-2}+\cdots+h_{j,n_{j}} \chi_{j}^{-n_{j}}+\cdots, \quad z \rightarrow \varphi_{j} .\label{hhexp}
\end{align}
Denote $\mathbf{u}=\mathbf{t}\cup\mathbf{h}$, with
	\begin{align}
		\mathbf{t}&=\{t_{i,s}|1\le i\le m,\ s\in \mathbb{Z}\}, \label{flatt}\\
		\mathbf{h}&=\{h_{0,j}|1\le j\le n_{0}-1\}\cup\{h_{k,r}|1\le k\le m,\ 0\le r\le n_{k}\}. \label{flath}
	\end{align}

\begin{prop} \label{etaflat}
The metric \eqref{whimetric} on the manifold \(\mathcal{M}\) is flat, with a system of  flat coordinates given by $\mathbf{u}$.
More exactly, it holds that
	\begin{align}
		&\left\langle\frac{\partial}{\partial t_{i,s}}, \frac{\partial}{\partial t_{i,s'}}\right\rangle_\eta=-d_{i} \delta_{-d_{i}, s+s'}, \quad i\in\left\{1,\cdots,m\right\},\ s,s'\in\mathbb{Z}; \label{etaflat1}\\
		&\left\langle\frac{\partial}{\partial h_{0,j}}, \frac{\partial}{\partial h_{0,j'}}\right\rangle_\eta=n_{0} \delta_{n_{0}, j+j'},\quad  j,j'\in\left\{1,\cdots,n_{0}-1\right\};\\
		&\left\langle\frac{\partial}{\partial h_{k,r}}, \frac{\partial}{\partial h_{k,r'}}\right\rangle_\eta=n_{k} \delta_{n_{k}, r+r'},\quad k\in\left\{1,\cdots,m\right\},\ r,r'\in\left\{0,\cdots,n_{k}\right\}, \label{etaflat3}
	\end{align}
and any other pairing between these vectors vanishes.
\end{prop}
\begin{proof}
It is easy to see:
	$$
	\begin{aligned}
		& \left.\frac{\partial \zeta(z)}{\partial t_{i,s}}\right|_{\gamma_{i'}}=-\dt_{i i'}\zeta(z)^{\frac{s}{d_{i}}} \zeta^{\prime}(z), 
\quad
\frac{\partial \ell(z)}{\partial t_{i,s}}=0, \\
		& \frac{\partial \zeta(z)}{\partial h_{0,j}}=0, \quad \frac{\partial \ell(z)}{\partial h_{0,j}}=\left(\ell^{\prime}(z) \chi_{0}(z)^{-j}\right)_{\infty,\,\ge 0}, \\
		& \frac{\partial \zeta(z)}{\partial h_{k,r}}=0, \quad \frac{\partial \ell(z)}{\partial h_{k,r}}=-\left(\ell^{\prime}(z) \chi_{k}(z)^{-r}\right)_{\varphi_{k},\,\le -1}.
	\end{aligned}
	$$
Then the proposition is verified case by case as for Proposition 2.6 in \cite{ma2021infinite}.
\end{proof}

The proof together with \eqref{zetaell2a} lead to the following result.
\begin{cor}\label{flatvect}
The following equalities hold true (recall $\mathbf{1}_{\gamma_j}$ given by \eqref{1fun}):
	\begin{align}
		&	\frac{\p \vec{a}}{\p t_{i,s}}=\left(-(\zeta(z)^{\frac{s}{d_{i}}}\zeta'(z)\mathbf{1}_{\gamma_{i}})_{-}, (\zeta(z)^{\frac{s}{d_{i}}}\zeta'(z)\mathbf{1}_{\gamma_{i}})_{+}\right),\quad 1\le i\le m, ~ s\in\mathbb{Z}; \label{flatcom1} \\
		&	\frac{\p \vec{a}}{\p h_{0,j}}=\left((\ell'(z)\chi_{0}(z)^{-j})_{\infty,\,\ge 0},(\ell'(z)\chi_{0}(z)^{-j})_{\infty,\,\ge 0}\right),\quad 1\le j\le n_{0}-1;\label{flatcom2}\\
		&	\frac{\p \vec{a}}{\p h_{k,r}}=\left(-(\ell'(z)\chi_{k}(z)^{-r})_{\varphi_{k},\,\le -1},-(\ell'(z)\chi_{k}(z)^{-r})_{\varphi_{k},\,\le -1}\right),\quad 1\le k\le m, ~ 0\le r\le n_{k}.\label{flatcom3}
	\end{align}
\end{cor}

\subsection{Cotangent spaces}
In contrast to the case $m=1$ studied previously in \cite{ma2021infinite}, for general $m\ge2$ we do not have an explicit expression for the cotangent spaces of $\mathcal{M}$. Instead, we will consider the cotangent spaces as certain quotient spaces.

Given $\vec{a}\in\mathcal{M}$, let us start with a linear map  \(\eta\) from \(\mathcal{H} \times \mathcal{H}\) to the tangent space \(T_{\vec{a}}\mathcal{M}\) defined as
	\begin{align}
		\eta(\omega(z), \hat{\omega}(z))=&\big(a'(z)(\omega(z)+\hat{\omega}(z))_{-}-(\omega(z) a'(z)+\hat{\omega}(z)\hat{a}'(z))_{-},\nn\\
		&-\hat{a}'(z)(\omega(z)+\hat{\omega}(z))_{+}+(\omega(z) a'(z)+\hat{\omega}(z)\hat{a}'(z))_{+}\big). \label{etaom} 
	\end{align}
\begin{lem}\label{whiinveta}
The map $\eta$ is surjective. More exactly, for any \( \vec{\xi} = (\xi(z), \hat{\xi}(z)) \in T_{\vec{a}}\mathcal{M} \), there is an element \( \vec{\omega} = (\omega(z), \hat{\omega}(z)) \in \mathcal{H} \times \mathcal{H} \), which is given by
	\begin{align*}
		&\omega(z)_{+} = -\left( \frac{\xi(z) - \hat{\xi}(z)}{\zeta'(z)} \right)_{+}, \\
		&\omega(z)_{-} = \left( \frac{\xi(z)}{a'(z)} - \left( \frac{\xi(z) - \hat{\xi}(z)}{\zeta'(z)} \right)_{-} \right)_{\infty, \, \geq -n_{0} + 1}, \\
		&\hat{\omega}(z)_{-} = \left( \frac{\xi(z) - \hat{\xi}(z)}{\zeta'(z)} \right)_{-}, \\
		&\left.\hat{\omega}(z)_{+}\right|_{D_{j}\setminus\{\varphi_j\}} = \left( -\frac{\hat{\xi}(z)}{\hat{a}'(z)} + \left( \frac{\xi(z) - \hat{\xi}(z)}{\zeta'(z)} \right)_{+} \right)_{\varphi_{j},\, \leq n_{j}},\quad j\in\{1,2,\cdots,m\},
	\end{align*}
satisfying \( \eta (\vec{\omega}) = \vec{\xi} \).
\end{lem}
\begin{proof}
	Denote \( (\sigma, \hat{\sigma})=\eta (\vec{\omega}) \). By using \eqref{etaom} we have
	\begin{align}
		\sigma_{+} &= \left(a'(z)[\omega(z) + \hat{\omega}(z)]_{-}\right)_{+} \nonumber\\
		&= \left(a'(z)([\omega(z) + \hat{\omega}(z)]_{-})_{\infty,\, \geq -n_{0}+1}\right)_{+} \nonumber\\
		&= \left(a'(z)\left(\frac{\xi(z)}{a'(z)}\right)_{\infty,\, \geq -n_{0}+1}\right)_{+} \nonumber\\
		&= \xi(z)_{+}, \label{sigmap}
\\
		\hat{\sigma}_{-} &= -\sum_{j=1}^{m}(\hat{a}'(z)[\omega(z) + \hat{\omega}(z)]_{+})_{\varphi_{j}, \leq -1} \nonumber\\
		&= -\sum_{j=1}^{m}(\hat{a}'(z)([\omega(z) + \hat{\omega}(z)]_{+})_{\varphi_{j}, \leq n_{j}})_{\varphi_{j}, \,\leq -1} \nonumber\\
		&= \sum_{j=1}^{m}\left(\hat{a}'(z)\left(\frac{\hat{\xi}(z)}{\hat{a}'(z)}\right)_{\varphi_{j}, \leq n_{j}}\right)_{\varphi_{j}, \,\leq -1} \nonumber\\
		&= \sum_{j=1}^{m}\hat{\xi}(z)_{\varphi_{j}, \leq -1}  \nonumber \\
		&= \hat{\xi}(z)_{-}, \label{hsigmam}
\\
		\sigma - \hat{\sigma} &= (a'(z) - \hat{a}'(z))(\hat{\omega}(z)_{-} - \omega(z)_{+}) \\
		&= \zeta'(z)(\hat{\omega}(z)_{-} - \omega(z)_{+}) \\
		&= \xi(z) - \hat{\xi}(z). \label{sigma-hsigma}
	\end{align}
Taking these three equalities together, one has
\[
\sigma_{-} = \xi(z)_{-}, \quad \hat{\sigma}_{+} = \hat{\xi}(z)_{+}.
\]
Thus \( \eta (\vec{\omega}) = \vec{\xi} \), and the lemma is proved.
\end{proof}

Let us consider a pairing between an element \( \vec{\omega} = (\omega(z), \hat{\omega}(z))  \in \mathcal{H} \times \mathcal{H} \) and a vector $X\in T_{\vec{a}}\mathcal{M}$ as
\begin{equation}\label{whipair}
	\langle \vec{\omega}, X \rangle = \frac{1}{2\pi \mathrm{i}} \sum_{s=1}^{m} \oint_{\gamma_s} \left( \omega(z) \partial_X a(z) + \hat{\omega}(z) \partial_X \hat{a}(z) \right) dz.
\end{equation}

\begin{lem} \label{pairtancotan}
For any \( \vec{\omega} \in \mathcal{H} \times \mathcal{H} \) and  $X\in T_{\vec{a}}\mathcal{M}$, it holds that
	\begin{equation}\label{pairmec}
		\langle\vec{\omega},X\rangle=\langle\eta (\vec{\omega}),X\rangle_{\eta},
	\end{equation}
where $\langle~,~\rangle_{\eta}$ is the metric given in \eqref{whimetric}.
\end{lem}

\begin{proof}
Denote $\eta(\vec{\omega})=(\xi_{1}(z),\hat{\xi}_{1}(z))$. From \eqref{sigmap}--\eqref{sigma-hsigma} it follows that
	\begin{align}
& 		\xi_{1}(z)-\hat{\xi}_{1}(z)=-\zeta_{1}'(z)(\omega(z)_{+}-\hat{\omega}(z)_{-}), \label{etalem1}
\\
		&\left(\frac{\xi_{1}(z)}{a'(z)}\right)_{\infty,\,\ge -n_{0}+1}=\left( (\omega(z)+\hat{\omega}(z))_{-} \right)_{\infty,\,\ge-n_{0}+1},\label{etalem2} \\
		&\left(\frac{\hat{\xi}_{1}(z)}{\hat{a}'(z)}\right)_{\varphi_{j},\,\le n_{j}}=-\left( (\omega(z)+\hat{\omega}(z))_{+} \right)_{\varphi_{j},\,\le n_{j}},\quad j\in\{1,2,\cdots,m\}.\label{etalem3}
	\end{align}
On the other hand, writing $ \p_{X}\vec{a}=(\xi_{2}(z),\hat{\xi}_{2}(z))$, we have
\begin{equation}\label{I123}
		\langle \vec{\omega},X\rangle= \frac{1}{2\pi\mathrm{i}}\dsum_{j=1}^{m}\oint_{\gamma_{j}}\left(\omega(z)\xi_{2}(z) +\hat{\omega}(z)\hat{\xi}_{2}(z)\right)dz
=\mathcal{I}_{1}+\mathcal{I}_{2}+\mathcal{I}_{3},
\end{equation}
	where
	$$
	\begin{aligned}
		\mathcal{I}_{1}&=\frac{1}{2\pi\mathrm{i}}\dsum_{j=1}^{m}\oint_{\gamma_{j}}\left(\omega(z)_{+}-\hat{\omega}(z)_{-}\right)\left(\xi_{2}(z)-\hat{\xi}_{2}(z)\right)dz,\\
		\mathcal{I}_{2}&=\frac{1}{2\pi\mathrm{i}}\dsum_{j=1}^{m}\oint_{\gamma_{j}}\left(\omega(z)_{-}+\hat{\omega}(z)_{-}\right)\xi_{2}(z)dz,\\
		\mathcal{I}_{3}&=\frac{1}{2\pi\mathrm{i}}\dsum_{j=1}^{m}\oint_{\gamma_{j}}\left(\omega\left(z\right)_{+}+\hat{\omega}\left(z\right)_{+}\right)\hat{\xi}_{2}\left(z\right)dz.
	\end{aligned}
	$$
One rewrites $(\xi_{\nu}(z),\hat{\xi}_{\nu}(z))=\partial_\nu(a(z),\hat{a}(z))$ with $\nu=1,2$. With the help of \eqref{etalem1}, we have
	$$
\begin{aligned}
		\mathcal{I}_{1}&=-\frac{1}{2\pi\mathrm{i}} \dsum_{j=1}^{m}\oint_{\gamma_{j}}\frac{(\xi_{1}(z)-\hat{\xi}_{1}(z))(\xi_{2}(z)-\hat{\xi}_{2}(z))}{\zeta'(z)}d z
		=-\frac{1}{2\pi\mathrm{i}}\dsum_{j=1}^{m}\oint_{\gamma_{j}}\frac{\p_{1}\zeta(z)\cdot \p_{2}\zeta(z)}{\zeta'(z)}d z.
	\end{aligned}
	$$
Similarly, by using  \eqref{etalem2} and \eqref{etalem3}, we obtain
	\begin{align*}
		\mathcal{I}_{2}&=\frac{1}{2\pi\mathrm{i}}\dsum_{j=1}^{m}\oint_{\gamma_{j}} \left( [\omega(z)+\hat{\omega}(z)]_{-} \right)_{\infty,\,\ge-n_{0}+1}\xi_{2}(z)\,dz \\
		&=\frac{1}{2\pi\mathrm{i}}\dsum_{j=1}^{m}\oint_{\gamma_{j}} \left(\frac{\p_1 a(z)}{a'(z)}\right)_{\infty,\,\ge-n_{0}+1}\p_2 a(z)\,dz \\
		&=\frac{1}{2\pi\mathrm{i}}\dsum_{j=1}^{m}\oint_{\gamma_{j}} \left(\frac{\p_1 \ell(z)}{a'(z)}\right)_{\infty,\,\ge-n_{0}+1}\p_2 \ell(z) \,dz \\
		&=\frac{1}{2\pi\mathrm{i}}\dsum_{j=1}^{m}\oint_{\gamma_{j}} \left(\frac{\p_1 \ell(z)}{\ell'(z)}\right)_{\infty,\,\ge-n_{0}+1}\p_2 \ell(z) \,dz \\
		&=-\res_{z=\infty}\frac{\p_{1}\ell(z)\cdot \p_{2}\ell(z)}{\ell'(z)}\,dz,
\\
		\mathcal{I}_{3}&=\frac{1}{2\pi\mathrm{i}}\sum_{j=1}^{m}\oint_{\gamma_{j}} \left([\omega(z)+\hat{\omega}(z)]_{+}\right)_{\varphi_{j},\,\le n_{j}} \hat{\xi}_{2}(z)\,dz\\
		&=-\frac{1}{2\pi\mathrm{i}}\sum_{j=1}^{m}\oint_{\gamma_{j}} \left(\frac{\p_1\hat{a}(z)}{\hat{a}'(z)}\right)_{\varphi_{j},\,\le n_{j}} \p_2\hat{a}(z)\,dz\\
		&=-\frac{1}{2\pi\mathrm{i}}\sum_{j=1}^{m}\oint_{\gamma_{j}} \left(\frac{\p_1\ell(z)}{\hat{a}'(z)}\right)_{\varphi_{j},\,\le n_{j}} \p_2\ell(z)\,dz\\
		&=-\frac{1}{2\pi\mathrm{i}}\sum_{j=1}^{m}\oint_{\gamma_{j}} \left(\frac{\p_1\ell(z)}{\ell'(z)}\right)_{\varphi_{j},\,\le n_{j}} \p_2\ell(z)\,dz\\
		&=-\sum_{j=1}^{m}\res_{z=\varphi_{j}}\frac{\p_{1}\ell(z)\cdot \p_{2}\ell(z)}{\ell'(z)}\,dz.
	\end{align*}
Substituting them into \eqref{I123} and taking \eqref{whimetric} together, the lemma is proved.
\end{proof}

As an application of Lemmas \ref{whiinveta} and \ref{pairtancotan}, from now on we identify the cotangent space and a quotient space, namely, we take
\begin{equation}\label{cotspace}
 T^{\ast}_{\vec{a}}\mathcal{M}=( \mathcal{H} \times \mathcal{H} )/\mathrm{ker}\,\eta.
\end{equation}
The pairing between a covector and a vector is given by \eqref{whipair}, which is nondegenerate since the metric \eqref{whimetric} is nondegenerate.

%
%

\subsection{Frobenius algebra of the tangent space}
We want to construct a Frobenius algebra structure in each tangent space of $\mathcal{M} $, namely, a multiplication between vectors that is commutative, associative, with unit vector, and invariant with respect to the flat metric \eqref{whimetric}.

Firstly, given \( \vec{a}=(a(z),\hat{a}(z)) \in \mathcal{M} \), let us introduce a multiplication between two elements $\vec{\omega}_\nu=(\omega_\nu(z),\hat{\omega}_\nu(z))\in \mathcal{H} \times \mathcal{H}$ with $\nu\in\{1,2\}$ as
 \begin{align}
		\vec{\omega}_1 \circ \vec{\omega}_{2} =& (\omega_{2}(\omega_{1}a')_{+} - \omega_{1}(\omega_{2}a')_{-} - \omega_{2}(\hat{\omega}_{1}\hat{a}')_{-} - \omega_{1}(\hat{\omega}_{2}\hat{a}')_{-}, \nonumber\\
		& \hat{\omega}_{2}(\hat{\omega}_{1}\hat{a}')_{+} - \hat{\omega}_{1}(\hat{\omega}_{2}\hat{a}')_{-} + \hat{\omega}_{1}(\omega_{2}a')_{+} + \hat{\omega}_{2}(\omega_{1}a')_{+}). \label{multiHH}
	\end{align}
\begin{lem}\label{copro}
	The space \( \mathcal{H} \times \mathcal{H} \) is a commutative associative algebra with respect to the multiplication \eqref{multiHH}.
\end{lem}
\begin{proof}
	The commutativity of the multiplication is clear. We proceed to show the associativity. It suffices to verify the cases that each vector has a vanishing component. For instance, suppose that \( \vec{\omega}_{\nu} = (\omega_{\nu}, 0) \) with \( \nu \in\{1, 2, 3\} \).
It is straightforward to calculate
	\begin{align*}
		&(\vec{\omega}_{2} \circ \vec{\omega}_{1}) \circ \vec{\omega}_{3} \\
		=& (((\omega_{1}a')_{+}\omega_{2} - (\omega_{2}a')_{-}\omega_{1})a')_{+}\omega_{3} - (\omega_{3}a')_{-}((\omega_{1}a')_{+}\omega_{2} - (\omega_{3}a')_{-}\omega_{1}), 0) \\
		\sim& (((\omega_{1}a')_{+}\omega_{2}a')_{+}\omega_{3} - ((\omega_{2}a')_{-}\omega_{1}a')_{+}\omega_{3} - (\omega_{3}a')_{-}(\omega_{1}a')_{+}\omega_{2} + (\omega_{3}a')_{-}(\omega_{2}a')_{-}\omega_{1}, 0) \\
		=& (((\omega_{1}a')_{+}(\omega_{2}a')_{+})_{+}\omega_{3} - (\omega_{3}a')_{-}(\omega_{1}a')_{+}\omega_{2}, 0) \\
		\sim& (((\omega_{1}a')_{+}(\omega_{2}a')_{+})_{+}\omega_{3} + (\omega_{3}a')_{+}(\omega_{1}a')_{+}\omega_{2}, 0) \\
		\sim& (0, 0),
	\end{align*}
	where the notation ``\( f \sim g \)'' means that \( f - g \) is symmetric with respect to $\omega_{2}$ and $\omega_{3}$. Thus we obtain
\[ (\vec{\omega}_{2} \circ \vec{\omega}_{1}) \circ \vec{\omega}_{3} = \vec{\omega}_{2} \circ (\vec{\omega}_{1} \circ \vec{\omega}_{3}).
\]
%
	The other cases are similar. Therefore the lemma is proved.
\end{proof}

\begin{lem} \label{cotanmult}
The formula \eqref{multiHH} defines a commutative associative multiplication on the cotangent space
$ T^{\ast}_{\vec{a}}\mathcal{M}$ (recall \eqref{cotspace}).
\end{lem}
\begin{proof}
Thanks to Lemma~\ref{copro}, for any $\vec{\omega}_1, \vec{\omega}_2 \in \mathcal{H} \times \mathcal{H}$, we only need to show $\eta(\vec{\omega}_1\circ \vec{\omega}_2)=0$ whenever $\eta(\vec{\omega}_2)=0$. To this end,  for an arbitrary element $\vec{\omega}=(\omega(z),\hat{\omega}(z))\in \mathcal{H} \times \mathcal{H}$, let us write
\[
\mathrm{tr}(\vec{\omega})=\frac{1}{2\pi\mathrm{i}}\sum_{j=1}^{m} \oint_{\gamma_{j}}(\omega(z)+\hat{\omega}(z))dz.
\]
Writing $\vec{\omega}_\nu=(\omega_\nu(z),\hat{\omega}_\nu(z))$ with $\nu\in\{1,2\}$, one has
\begin{align}
\mathrm{tr}(\vec{\omega}_{1}\circ \vec{\omega}_2)=& \frac{1}{2\pi\mathrm{i}}\sum_{j=1}^{m} \oint_{\gamma_{j}}
\Big(  \omega_{2}(\omega_{1} a')_{+} - \omega_{1} (\omega_{2}a')_{-} - \omega_{2}(\hat{\omega}_{1} \hat{a}')_{-} - \omega_{1} (\hat{\omega}_{2}\hat{a}')_{-} \nonumber\\
		& \qquad\qquad + \hat{\omega}_{2}(\hat{\omega}_{1} \hat{a}')_{+} - \hat{\omega}_{1} (\hat{\omega}_{2}\hat{a}')_{-} + \hat{\omega}_{1} (\omega_{2}a')_{+} + \hat{\omega}_{2}(\omega_{1} a')_{+} \Big) dz
\nonumber
\\
=& \frac{1}{2\pi\mathrm{i}}\sum_{j=1}^{m} \oint_{\gamma_{j}}
\Big( \omega_{1}(-a'(\omega_2+\hat{\omega}_2)_- -(a'\omega_2+\hat{a}'\hat{\omega}_2)_- )  \nonumber\\
		& \qquad\qquad + \hat{\omega}_{1}(\hat{a}'(\omega_2+\hat{\omega}_2)_+ + (a'\omega_2+\hat{a}'\hat{\omega}_2)_+ )   \Big) dz
\nonumber
\\
=& \langle\vec{\omega}_{1},\eta(\vec{\omega}_{2})\rangle. \label{trprod}
\end{align}
Hence
\[
\langle\vec{\omega},\eta(\vec{\omega}_{1}\circ\vec{\omega}_{2})\rangle =\mathrm{tr}(\vec{\omega}\circ\vec{\omega}_{1}\circ\vec{\omega}_{2}) =\langle\vec{\omega}\circ\vec{\omega}_{1},\eta(\vec{\omega}_{2})\rangle=0.
\]
Since $\vec{\omega}$ is arbitrary and the pairing is nondegenerate, then we obtain $\eta(\vec{\omega}_{1}\circ\vec{\omega}_{2})=0$. The lemma is proved.
\end{proof}

Let us introduce a multiplication on the tangent space \( T_{\vec{a}}\mathcal{M} \) as
\begin{equation}\label{whimul}
	\vec{\xi}_{1} \circ \vec{\xi}_{2} = \eta \circ \left(\eta^{-1}(\vec{\xi}_{1}) \circ \eta^{-1}(\vec{\xi}_{2}) \right).
\end{equation}
\begin{prop} \label{tanprod}
The multiplication \eqref{whimul} on \( T_{\vec{a}}\mathcal{M} \) is commutative, associate and invariant with respect to the metric \eqref{whimetric}.
\end{prop}
\begin{proof}
The commutativity and associativity follows from Lemmas~\ref{whiinveta} and \ref{cotanmult}.
For any $\vec{\xi}_{1}, \vec{\xi}_{2}, \vec{\xi}_{3}\in T_{\vec{a}}\mathcal{M}$, by using \eqref{pairmec} and \eqref{trprod} we have
\begin{align*}
\langle \vec{\xi}_{1} \circ \vec{\xi}_{2}, \vec{\xi}_{3}\rangle_\eta =&  \langle \eta^{-1}(\vec{\xi}_{1})\circ \eta^{-1} (\vec{\xi}_{2}), \eta (\eta^{-1}(\vec{\xi}_{3}))\rangle
\\
=& \mathrm{tr}\left(\eta^{-1}(\vec{\xi}_{1})\circ \eta^{-1} (\vec{\xi}_{2})\circ\eta^{-1}(\vec{\xi}_{3})\right)
\\
=&\langle \vec{\xi}_{1},  \vec{\xi}_{2}\circ \vec{\xi}_{3}\rangle_\eta,
\end{align*}
which confirms the invariance. The proposition is proved.
\end{proof}

Given \( \vec{\xi} = (\xi, \hat{\xi}) \in T_{\vec{a}}\mathcal{M} \), we consider a linear
map
\[
C_{\vec{\xi}}: \mathcal{H} \times \mathcal{H} \to  T_{\vec{a}}\mathcal{M}
\]
defined by
	\begin{equation}\label{whiprod}
		C_{\vec{\xi}} (\omega, \hat{\omega}) = \left(a'(\omega \xi + \hat{\omega} \hat{\xi})_{-} - \xi (\omega a' + \hat{\omega} \hat{a}')_{-}, -\hat{a}'(\omega \xi + \hat{\omega} \hat{\xi})_{+} + \hat{\xi} (\omega a' + \hat{\omega} \hat{a}')_{+}\right).
	\end{equation}
\begin{cor}\label{Cxieta}
	The linear map \eqref{whiprod} satisfies
	\begin{equation}\label{Cxiom}
C_{\vec{\xi}} (\vec{\omega})=\vec{\xi}\circ\eta(\vec{\omega}), \quad \vec{\omega}\in \mathcal{H}\times \mathcal{H}.
	\end{equation}
\end{cor}
\begin{proof}
For any $\vec{\omega}_1, \vec{\omega}_2 \in \mathcal{H}\times \mathcal{H}$, it is straightforward to check
\begin{equation}\label{eq-Cxi}
\langle \vec{\omega}_1 \circ \vec{\omega}_2, \vec{\xi} \rangle = \langle \vec{\omega}_1, C_{\vec{\xi}} (\vec{\omega}_2) \rangle.
\end{equation}
By using Lemma \ref{pairtancotan} and Proposition~\ref{tanprod}, we have
\begin{align*}
\langle \vec{\omega}_1 \circ \vec{\omega}_2, \vec{\xi} \rangle=&\langle \eta(\vec{\omega}_1 \circ \vec{\omega}_2), \vec{\xi} \rangle_\eta \\
=&  \langle \eta^{-1}(\vec{\omega}_1) \circ \eta^{-1}(\vec{\omega}_2), \vec{\xi} \rangle_\eta
\\
=&  \langle \eta^{-1}(\vec{\omega}_1) , \vec{\xi}\circ \eta^{-1}(\vec{\omega}_2) \rangle_\eta
\\
=&  \langle \vec{\omega}_1 , \vec{\xi}\circ \eta^{-1}(\vec{\omega}_2) \rangle.
\end{align*}
Taking \eqref{eq-Cxi}  together and noting that $\vec{\omega}_1$ is arbitrary, we conclude the corollary.
\end{proof}


\begin{prop}
In the algebra \( (T_{\vec{a}}\mathcal{M}, \circ) \) there is a unit vector  given by (recall \eqref{flatt} and \eqref{flath})
	\begin{equation}\label{whiid}
		e  =
		\begin{cases}
			\sum_{i=1}^{m}\left(\frac{\partial}{\partial t_{i,0}} + \frac{\partial}{\partial h_{i,0}}\right), & \quad  n_{0} = 1; \\
			\frac{1}{n_{0}} \frac{\partial}{\partial h_{0,n_{0}-1}}, & \quad n_{0} \geq 2.
		\end{cases}
	\end{equation}
\end{prop}
\begin{proof}
By using Corollary~\ref{flatvect}, we have
\begin{equation}
	\p_{e}\vec{a}
	=\left\{
	\begin{aligned}
		&(1-a'(z),1-\hat{a}'(z)),\quad &n_{0}&=1;\\
		&(1,1),\quad &n_{0}&\ge2.\\
	\end{aligned}
	\right.
\end{equation}
It is straightforward to verify that 	
	$$
	C_{e}( \vec{\omega} )=\eta(\vec{\omega}),\quad \vec{\omega}\in\mathcal{H}\times\mathcal{H},
	$$
Hence, with the help of \eqref{Cxiom}, we obtain
		$$
	e\circ\vec{\xi}=C_{e}( \eta^{-1}(\vec{\xi}))=\vec{\xi},\quad \vec{\xi}\in T_{\vec{a}}\mathcal{M}.
	$$
	The proposition is proved.
\end{proof}

\subsection{Symmetric $3$-tensor}
Now let us consider the following symmetric $3$-tensor
 \begin{equation}\label{3tensor}
 c(\vec{\xi}_1, \vec{\xi}_2, \vec{\xi}_3) = \langle \vec{\xi}_1 \circ \vec{\xi}_2, \vec{\xi}_3 \rangle_{\eta}, 
 \quad \vec{\xi}_\nu \in T_{\vec{a}}\mathcal{M}.
 \end{equation}
\begin{prop}\label{3tensorflat}
The symmetric $3$-tensor defined by  \( c \) can be represented in the flat coordinates \eqref{flatt}--\eqref{flath} as follows:
{\small
		\begin{align*}
		c(\p_{t_{i_1,s_{1}}},\p_{t_{i_2,s_2}},\p_{t_{i_3,s_{3}}})=&
		\frac{1}{2\pi \mathrm{i}}\oint_{\gamma_{i_{1}}}\left(-\delta_{i_{1},i_{2}}\delta_{i_{2},i_{3}} \zeta^{\frac{s_{1}+s_{2}+s_{3}}{d_{i_{1}}}}\zeta'\mathbf{1}_{\gamma_{i_{1}}}(\zeta'+\zeta'_{-}+\ell')dz\right)\\
		& +	\frac{1}{2\pi \mathrm{i}}\dsum_{s=1}^{m}\oint_{\gamma_{s}}\bigg(\delta_{i_{1},i_{2}}\zeta^{\frac{s_{1}
				+s_{2}}{d_{i_{1}}}}\zeta'\mathbf{1}_{\gamma_{i_{1}}}(\zeta^{\frac{s_{3}}{d_{i_{3}}}}\zeta'\mathbf{1}_{i_{3}})_{-}
\\
& +\mathrm{c.p.}\{(i_{1},s_{1}),(i_{2},s_{2}),(i_{3},s_{3})\}\bigg)d z,\nn\\
		c(\p_{h_{0,j_1}},\p_{h_{0,j_2}},\p_{h_{0,j_3}})=&-\res_{z=\infty}\left(-\chi_{0}^{-j_{1}-j_{2}-j_{3}}\ell'(2\ell'+\zeta'_{-})\right)dz\\
		&-\res_{z=\infty}\bigg(
		\chi_{0}^{-j_{1}-j_{2}}\ell'(\ell'\chi_{0}^{-j_{3}})_{\infty,\,\ge 0}
+\mathrm{c.p.}\{j_{1},j_{2},j_{3}\}\bigg)dz,\nn\\
		c(\p_{h_{k_1,r_{1}}},\p_{h_{k_2,r_{2}}}, \p_{h_{k_3,r_{3}}})=&\res_{z=\varphi_{k_{1}}}\left(-\delta_{k_{1},k_{2}}\delta_{k_{2},k_{3}}\chi_{k_{1}}^{-r_{1}-r_{2}-r_{3}}\ell'(2\ell'-\zeta'_{+})\right)dz\\
		&+\sum_{s=1}^{m}\res_{z=\varphi_{s}}\bigg( \delta_{k_{1},k_{2}}\chi_{k_{1}}^{-r_{1}-r_{2}}\ell'(\ell'\chi_{k_{3}}^{-r_{3}})_{\varphi_{k_{3},\,\le -1}}
\\
&
+\mathrm{c.p.}\{(k_{1},r_{1}),(k_{2},r_{2}),(k_{3},r_{3})\}\bigg)d z,\nn\\
		c({\p_{t_{i_1,s_{1}}}},{\p_{t_{i_2,s_{2}}}},{\p_{h_{0,j_3}}})=&-\frac{1}{2\pi \mathrm{i}}\oint_{\gamma_{i_{1}}}\delta_{i_{1},i_{2}}\zeta^{\frac{s_{1}+s_{2}}{d_{i_{1}}}}\zeta'(\ell'\chi_{0}^{-j_{3}})_{\infty,\,\ge 0}dz, \nn\\
		c(\p_{t_{i_1,s_{1}}},\p_{t_{i_2,s_{2}}},\p_{h_{k_{3},r_{3}}})=&
		\frac{1}{2\pi \mathrm{i}}\oint_{\gamma_{i_{1}}}\delta_{i_{1},i_{2}}\zeta^{\frac{s_{1}+s_{2}}{d_{i_{1}}}}\zeta'(\ell'\chi_{k_{3}}^{-r_{3}})_{\varphi_{k_{3}},\,\le -1}dz, \nn\\ c(\p_{h_{0,j_1}},\p_{h_{0,j_2}},\p_{t_{i_{3},s_{3}}})=&
		\res_{z=\infty}\chi_{0}^{-j_{1}-j_{2}}\ell'(\zeta^{\frac{s_{3}}{d_{i_{3}}}}\zeta'\mathbf{1}_{\gamma_{i_{3}}})_{-}dz,\nn \\
		c(\p_{h_{0,j_1}},\p_{h_{0,j_2}},\p_{h_{k_{3},r_{3}}})=&
		\res_{z=\infty}\chi_{0}^{-j_{1}-j_{2}}\ell'(\ell'\chi_{k_{3}}^{-r_{3}})_{\varphi_{k_{3}},\,\le -1 }dz,\nn \\ c(\p_{h_{k_1,r_{1}}},\p_{h_{k_2,r_{2}}},\p_{t_{i_{3},s_{3}}})=&
		-\res_{z=\varphi_{k_{1}}}\delta_{k_{1},k_{2}}\chi_{k_{1}}^{-r_{1}-r_{2}}\ell'(\zeta^{\frac{s_{3}}{d_{i_{3}}}}\zeta'\mathbf{1}_{\gamma_{i_{3}}})_{+}dz,\nn\\
		c(\p_{h_{k_1,r_{1}}},\p_{h_{k_2,r_{2}}},\p_{h_{0,j_{3}}})=&
		-\res_{z=\varphi_{k_{1}}}\delta_{k_{1},k_{2}}\chi_{k_{1}}^{-r_{1}-r_{2}}\ell'(\ell'\chi_{0}^{-j_{3}})_{\infty,\,\ge 0}dz,\nn\\
		c(\p_{t_{i_1,s_{1}}},\p_{h_{0,j_2}},\p_{h_{k_{3},r_{3}}})=&0.
\end{align*}}
Here  $\mathrm{c.p.}\{(i_{1},s_{1}),(i_{2},s_{2}),(i_{3},s_{3})\}$ denotes the sum over cyclic permutations of the index pairs.
\end{prop}
The proof of this proposition is similar as in \cite{ma2021infinite}.
\begin{lem}\label{dulbas}
Up to the kernel of $\eta$ given in \eqref{etaom}, it holds that:
	\begin{align}
	\eta^{-1} \left(\frac{\partial}{\partial t_{i,s}}\right) = (\zeta(z)^{\frac{s}{d_{i}}}\mathbf{1}_{\gamma_{i}}, -\zeta(z)^{\frac{s}{d_{i}}}\mathbf{1}_{\gamma_{i}}), &\quad s \in \mathbb{Z}; \\
	\eta^{-1}\left(\frac{\partial}{\partial h_{0,j}}\right) = ((\chi_{0}(z)^{-j})_{\infty,\,\ge -n_{0}+1}, 0), &\quad 1 \leq j \leq n_{0}-1; \\
	\eta^{-1}\left(\frac{\partial}{\partial h_{k,r}}\right) = (0, (\chi_{k}(z)^{-r})_{\varphi_{k},\,\le n_{k}}\mathbf{1}_{\gamma_{k}}), &\quad 0 \leq r \leq n_{k}.
	\end{align}
	\end{lem}
	\begin{proof}
One substitutes \eqref{flatcom1}--\eqref{flatcom3} into the formulae in Lemma~\ref{whiinveta}. Then, by using the following equalities:
		\begin{align*}
			&\left(\frac{\omega(z)_-}{a'(z)}\right)_{\infty,\,\ge-n_{0}+1}=0, \quad \left(\frac{\omega(z)_+}{\hat{a}'(z)}\right)_{\varphi_{k},\,\le n_{k}}=0 \quad \hbox{with arbitrary}
			\quad \omega(z)\in\mathcal{H}; \\
			&(\ell'(z)\chi_{0}(z)^{-j})_{\infty,\,\ge 0}=(a'(z)\chi_{0}(z)^{-j})_{\infty,\,\ge 0}, \\
&(\ell'(z)\hat{\chi}_{k}(z)^{-r})_{\varphi_{k},\,\le -1}=(\hat{a}'(z)\hat{\chi}_{k}(z)^{-r})_{\varphi_{k},\,\le -1},
		\end{align*}
the lemma is proved.
	\end{proof}

\begin{lem}\label{dulbaspro}
	The following equalities hold true:	
		\begin{align}
	\eta^{-1} \frac{\p}{\p t_{i_{1},s_{1}}}\circ \eta^{-1} \frac{\p}{\p t_{i_{2},s_{2}}}=&\bigg(\delta_{i_{1},i_{2}}\zeta^{\frac{s_{1}+s_{2}}{d_{i_{1}}}}a'\mathbf{1}_{\gamma_{i_{1}}}
	-\zeta^{\frac{s_{1}}{d_{i_{1}}}}(\zeta^{\frac{s_{2}}{d_{i_{2}}}} \zeta'\mathbf{1}_{\gamma_{i_{2}}})_{-}\mathbf{1}_{\gamma_{i_{1}}}	
\nn \\
	& -\zeta^{\frac{s_{2}}{d_{i_{2}}}}(\zeta^{\frac{s_{1}}{d_{i_{1}}}}\zeta'\mathbf{1}_{\gamma_{i_{1}}})_{-}\mathbf{1}_{\gamma_{i_{2}}} ,\, -\delta_{i_{1},i_{2}}\zeta^{\frac{s_{1}+s_{2}}{d_{i_{1}}}}\hat{a}'\mathbf{1}_{\gamma_{i_{1}}}
\nn \\
	& 	-\zeta^{\frac{s_{1}}{d_{i_{1}}}}(\zeta^{\frac{s_{2}}{d_{i_{2}}}}\zeta'\mathbf{1}_{\gamma_{i_{2}}})_{+} \mathbf{1}_{\gamma_{i_{1}}}	-\zeta^{\frac{s_{2}}{d_{i_{2}}}}(\zeta^{\frac{s_{1}}{d_{i_{1}}}}\zeta'\mathbf{1}_{\gamma_{i_{1}}})_{+} \mathbf{1}_{\gamma_{i_{2}}}\bigg),\nn
\\
\eta^{-1} \frac{\p}{\p h_{0,j_{1}}}\circ \eta^{-1} \frac{\p}{\p h_{0,j_{2}}}=&(-(\chi_{0}^{-j_{1}})_{\infty,\,\ge -n_{0}+1}(\chi_{0}^{-j_{2}})_{\infty,\,\ge -n_{0}+1}a'
\nn \\
	&
+(\chi_{0}^{-j_{1}})_{\infty,\,\ge -n_{0}+1}(\chi_{0}^{-j_{2}}\ell')_{\infty,\,\ge 0}\nn\\
	&+(\chi_{0}^{-j_{2}})_{\infty,\,\ge -n_{0}+1}(\chi_{0}^{-j_{1}}\ell')_{\infty,\,\ge 0},\,0),
\nn\\
	\eta^{-1} \frac{\p}{\p h_{k_{1},r_{1}}}\circ \eta^{-1} \frac{\p}{\p h_{k_{2},r_{2}}}=&(0,\delta_{k_{1},k_{2}}(\chi_{k_{1}}^{-r_{1}})_{\varphi_{k_{1}},\,\le n_{k_{1}}}(\chi_{k_{2}}^{-r_{2}})_{\varphi_{k_{2}},\,\le n_{k_{2}}}\hat{a}'\mathbf{1}_{\gamma_{k_{1}}}\nn\\
	&-(\chi_{k_{1}}^{-r_{1}})_{\varphi_{k_{1}},\,\le n_{k_{1}}}(\chi_{k_{2}}^{-r_{2}}\ell')_{\varphi_{k_{2}},\,\le -1}\mathbf{1}_{\gamma_{k_{1}}}
\nn \\
	&
-(\chi_{k_{2}}^{-r_{2}})_{\varphi_{k_{2}},\,\le n_{k_{2}}}(\chi_{k_{1}}^{-r_{1}}\ell')_{\varphi_{k_{1}},\,\le -1}\mathbf{1}_{\gamma_{k_{2}}}),\nn\\
	\eta^{-1} \frac{\p}{\p h_{0,j}}\circ \eta^{-1} \frac{\p}{\p h_{k,r}}=&(-(\chi_{0}^{-j})_{\infty,\,\ge -n_{0}+1}(\chi_{k}^{-r}\ell')_{\varphi_{k},\,\le -1},
	(\chi_{k}^{-r})_{\varphi_{k},\,\le n_{k}}(\chi_{0}^{-j}\ell')_{\infty,\,\ge 0}\mathbf{1}_{\gamma_{k}}).\nn
\end{align}
\end{lem}
\begin{proof}
	This lemma is verified by taking Lemmas~\ref{copro} and \ref{dulbas} together, with the help of the following equalities:
	$$
	\begin{aligned}
		&((\chi_{0}^{-j})_{\infty,\,\ge -n_{0}+1}a')_{+}=(\chi_{0}^{-j} \ell')_{\infty,\,\ge 0},\quad ((\chi_{k}^{-r})_{\varphi_{k},\,\le n_{k}}\hat{a}'\mathbf{1}_{\gamma_{k}})_{-}=(\chi_{k}^{-r}\ell')_{\varphi_{k},\,\le -1}.
	\end{aligned}
	$$
	For instance, one has
	$$
	\begin{aligned}
		\eta^{-1} \frac{\p}{\p t_{i_{1},s_{1}}}\circ \eta^{-1} \frac{\p}{\p t_{i_{2},s_{2}}}&=\left(\zeta(z)^{\frac{s_{1}}{d_{i_{1}}}}\mathbf{1}_{\gamma_{i_{1}}}, -\zeta(z)^{\frac{s_{1}}{d_{i_{1}}}}\mathbf{1}_{\gamma_{i_{1}}}\right)\circ\left(\zeta(z)^{\frac{s_{2}}{d_{i_{2}}}}\mathbf{1}_{\gamma_{i_{2}}}, -\zeta(z)^{\frac{s_{2}}{d_{i_{2}}}}\mathbf{1}_{\gamma_{i_{2}}}\right)\\
		&=\left(\om,\hat{\om} \right),
	\end{aligned}
	$$
	where
	\begin{align*}
		\om=&\left(\zeta^{\frac{s_{2}}{d_{i_{2}}}}\mathbf{1}_{\gamma_{i_{2}}}\right)\left(\zeta^{\frac{s_{1}}{d_{i_{1}}}}a'\mathbf{1}_{\gamma_{i_{1}}}\right)_{+}-\zeta^{\frac{s_{1}}{d_{i_{1}}}}\mathbf{1}_{\gamma_{i_{1}}}\left(\zeta^{\frac{s_{2}}{d_{i_{2}}}}a'\mathbf{1}_{\gamma_{i_{2}}}\right)_{-} \\
		& +\left(\zeta^{\frac{s_{2}}{d_{i_{2}}}}\mathbf{1}_{\gamma_{i_{2}}}\right)\left(\zeta^{\frac{s_{1}}{d_{i_{1}}}}\hat{a}'\mathbf{1}_{\gamma_{i_{1}}}\right)_{-}+\zeta^{\frac{s_{1}}{d_{i_{1}}}}\mathbf{1}_{\gamma_{i_{1}}}\left(\zeta^{\frac{s_{2}}{d_{i_{2}}}}\hat{a}'\mathbf{1}_{\gamma_{i_{2}}}\right)_{-}\\
		=&\delta_{i_{1},i_{2}}\zeta^{\frac{s_{1}+s_{2}}{d_{i_{1}}}}a'\mathbf{1}_{\gamma_{i_{1}}}-\left(\zeta^{\frac{s_{2}}{d_{i_{2}}}}\mathbf{1}_{\gamma_{i_{2}}}\right)\left(\zeta^{\frac{s_{1}}{d_{i_{1}}}}a'\mathbf{1}_{\gamma_{i_{1}}}\right)_{-}-\zeta^{\frac{s_{1}}{d_{i_{1}}}}\mathbf{1}_{\gamma_{i_{1}}}\left(\zeta^{\frac{s_{2}}{d_{i_{2}}}}a'\mathbf{1}_{\gamma_{i_{2}}}\right)_{-} \\
		& +\left(\zeta^{\frac{s_{2}}{d_{i_{2}}}}\mathbf{1}_{\gamma_{i_{2}}}\right)\left(\zeta^{\frac{s_{1}}{d_{i_{1}}}}\hat{a}'\mathbf{1}_{\gamma_{i_{1}}}\right)_{-}+\zeta^{\frac{s_{1}}{d_{i_{1}}}}\mathbf{1}_{\gamma_{i_{1}}}\left(\zeta^{\frac{s_{2}}{d_{i_{2}}}}\hat{a}'\mathbf{1}_{\gamma_{i_{2}}}\right)_{-}\\
		=&\delta_{i_{1},i_{2}}\zeta^{\frac{s_{1}+s_{2}}{d_{i_{1}}}}a'\mathbf{1}_{\gamma_{i_{1}}}
		-\zeta^{\frac{s_{1}}{d_{i_{1}}}}\left(\zeta^{\frac{s_{2}}{d_{i_{2}}}}\zeta'\mathbf{1}_{\gamma_{i_{2}}}\right)_{-}\mathbf{1}_{\gamma_{i_{1}}}	-\zeta^{\frac{s_{2}}{d_{i_{2}}}}\left(\zeta^{\frac{s_{1}}{d_{i_{1}}}}\zeta'\mathbf{1}_{\gamma_{i_{1}}}\right)_{-}\mathbf{1}_{\gamma_{i_{2}}},
\\ \hat{\om}=&\left(\zeta^{\frac{s_{2}}{d_{i_{2}}}}\mathbf{1}_{\gamma_{i_{1}}}\right)\left(\zeta^{\frac{s_{1}}{d_{i_{1}}}}\hat{a}'\mathbf{1}_{\gamma_{i_{1}}}\right)_{+}-\left(\zeta^{\frac{s_{1}}{d_{i_{1}}}}\mathbf{1}_{\gamma_{i_{1}}}\right)\left(\zeta^{\frac{s_{2}}{d_{i_{2}}}}\hat{a}'\mathbf{1}_{\gamma_{i_{2}}}\right)_{-}\\
		&-\left(\zeta^{\frac{s_{2}}{d_{i_{2}}}}\mathbf{1}_{\gamma_{i_{2}}}\right)\left(\zeta^{\frac{s_{1}}{d_{i_{1}}}}a'\mathbf{1}_{\gamma_{i_{1}}}\right)_{+}-\left(\zeta^{\frac{s_{1}}{d_{i_{1}}}}\mathbf{1}_{\gamma_{i_{1}}}\right)\left(\zeta^{\frac{s_{2}}{d_{i_{2}}}}a'\mathbf{1}_{\gamma_{i_{2}}}\right)_{+} \\
	=&\left(\zeta^{\frac{s_{2}}{d_{i_{2}}}}\mathbf{1}_{\gamma_{i_{1}}}\right)\left(\zeta^{\frac{s_{1}}{d_{i_{1}}}}\hat{a}'\mathbf{1}_{\gamma_{i_{1}}}\right)_{+}-\delta_{i_{1},i_{2}}\zeta^{\frac{s_{1}+s_{2}}{d_{i_{1}}}}\hat{a}'\mathbf{1}_{\gamma_{i_{1}}}+\left(\zeta^{\frac{s_{1}}{d_{i_{1}}}}\mathbf{1}_{\gamma_{i_{1}}}\right)\left(\zeta^{\frac{s_{2}}{d_{i_{2}}}}\hat{a}'\mathbf{1}_{\gamma_{i_{2}}}\right)_{+}\\
	&-\left(\zeta^{\frac{s_{2}}{d_{i_{2}}}}\mathbf{1}_{\gamma_{i_{2}}}\right)\left(\zeta^{\frac{s_{1}}{d_{i_{1}}}}a'\mathbf{1}_{\gamma_{i_{1}}}\right)_{+}-\left(\zeta^{\frac{s_{1}}{d_{i_{1}}}}\mathbf{1}_{\gamma_{i_{1}}}\right)\left(\zeta^{\frac{s_{2}}{d_{i_{2}}}}a'\mathbf{1}_{\gamma_{i_{2}}}\right)_{+} \\
	=&-\delta_{i_{1},i_{2}}\zeta^{\frac{s_{1}+s_{2}}{d_{i_{1}}}}\hat{a}'\mathbf{1}_{\gamma_{i_{1}}}
	-\zeta^{\frac{s_{1}}{d_{i_{1}}}}\left(\zeta^{\frac{s_{2}}{d_{i_{2}}}}\zeta'\mathbf{1}_{\gamma_{i_{2}}}\right)_{+}\mathbf{1}_{\gamma_{i_{1}}}	-\zeta^{\frac{s_{2}}{d_{i_{2}}}}\left(\zeta^{\frac{s_{1}}{d_{i_{1}}}}\zeta'\mathbf{1}_{\gamma_{i_{1}}}\right)_{+}\mathbf{1}_{\gamma_{i_{2}}}.
	\end{align*}
	The other cases are similar. Thus the lemma is proved.
\end{proof}

\begin{proof}[Proof of Proposition~\ref{3tensorflat}]
Since
	\begin{equation}\label{3tensor2}
		c(\p_{u},\p_{v},\p_{w})=\left\la\p_u\circ\p_v, \p_w\right\ra_\eta= \left\langle \eta^{-1}\p_u\circ \eta^{-1}\p_v, \p_w\right\rangle, \quad u, v,
		w\in  \mathbf{u},
	\end{equation}
then the results follow from a straightforward calculation, with the help of Corollary~\ref{flatvect} and Lemma~\ref{dulbaspro}. For instance, we have
	\begin{align*}
		&c(\p_{t_{i_{1},s_{1}}},\p_{t_{i_{2},s_{2}}},\p_{t_{i_{3},s_{3}}})= \left\langle \frac{\p}{\p t_{i_{1},s_{1}}}\circ \frac{\p}{\p t_{i_{2},s_{2}}},
		\frac{\p}{\p t_{i_{3},s_{3}}}\right\rangle_{\eta} \\
		=&\frac{1}{2\pi\mathrm{i}}\sum_{s=1}^{m}\oint_{\gamma_{s}}
		\left([\delta_{i_{1},i_{2}}\zeta^{\frac{s_{1}+s_{2}}{d_{i_{1}}}}a'\mathbf{1}_{\gamma_{i_{1}}}
		-\zeta^{\frac{s_{1}}{d_{i_{1}}}}(\zeta^{\frac{s_{2}}{d_{i_{2}}}}\zeta'\mathbf{1}_{\gamma_{i_{2}}})_{-}\mathbf{1}_{\gamma_{i_{1}}}	-\zeta^{\frac{s_{2}}{d_{i_{2}}}}(\zeta^{\frac{s_{1}}{d_{i_{1}}}}\zeta'\mathbf{1}_{\gamma_{i_{1}}})_{-}\mathbf{1}_{\gamma_{i_{2}}}]
		(-\zeta^{\frac{s_{3}}{d_{i_{3}}}}\zeta'\mathbf{1}_{\gamma_{i_{3}}})_{-} \right.\nn \\
		&\qquad \left. +[-\delta_{i_{1},i_{2}}\zeta^{\frac{s_{1}+s_{2}}{d_{i_{1}}}}\hat{a}'\mathbf{1}_{\gamma_{i_{1}}}
		-\zeta^{\frac{s_{1}}{d_{i_{1}}}}(\zeta^{\frac{s_{2}}{d_{i_{2}}}}\zeta'\mathbf{1}_{\gamma_{i_{2}}})_{+}\mathbf{1}_{\gamma_{i_{1}}}	-\zeta^{\frac{s_{2}}{d_{i_{2}}}}(\zeta^{\frac{s_{1}}{d_{i_{1}}}}\zeta'\mathbf{1}_{\gamma_{i_{1}}})_{+}\mathbf{1}_{\gamma_{i_{2}}}]
		(\zeta^{\frac{s_{3}}{d_{i_{3}}}}\zeta'\mathbf{1}_{\gamma_{i_{3}}})_{+}\right) d z\\
			=&\frac{1}{2\pi\mathrm{i}}\sum_{s=1}^{m}\oint_{\gamma_{s}}
		\left([-\delta_{i_{1},i_{2}}\zeta^{\frac{s_{1}+s_{2}}{d_{i_{1}}}}a'\mathbf{1}_{\gamma_{i_{1}}}
		+\zeta^{\frac{s_{1}}{d_{i_{1}}}}(\zeta^{\frac{s_{2}}{d_{i_{2}}}}\zeta'\mathbf{1}_{\gamma_{i_{2}}})_{-}\mathbf{1}_{\gamma_{i_{1}}}	+\zeta^{\frac{s_{2}}{d_{i_{2}}}}(\zeta^{\frac{s_{1}}{d_{i_{1}}}}\zeta'\mathbf{1}_{\gamma_{i_{1}}})_{-}\mathbf{1}_{\gamma_{i_{2}}}]
		(\zeta^{\frac{s_{3}}{d_{i_{3}}}}\zeta'\mathbf{1}_{\gamma_{i_{3}}}) \right.\nn \\
		&\qquad \left. +[-\delta_{i_{1},i_{2}}\zeta^{\frac{s_{1}+s_{2}}{d_{i_{1}}}}(\hat{a}'-a')\mathbf{1}_{\gamma_{i_{1}}}
		-2\delta_{i_{1},i_{2}}\zeta^{\frac{s_{1}+s_{2}}{d_{i_{1}}}}\zeta'\mathbf{1}_{\gamma_{i_{1}}}	]
		(\zeta^{\frac{s_{3}}{d_{i_{3}}}}\zeta'\mathbf{1}_{\gamma_{i_{3}}})_{+}\right) d z\\
		=&\frac{1}{2\pi\mathrm{i}}\sum_{s=1}^{m}\oint_{\gamma_{s}}
		\left(-\delta_{i_{1},i_{2}}\delta_{i_{2},i_{3}}\zeta^{\frac{s_{1}+s_{2}+s_3}{d_{i_{1}}}}a'\mathbf{1}_{\gamma_{i_{1}}}+\delta_{i_{1},i_{3}}\zeta^{\frac{s_{1}+s_{3}}{d_{i_{1}}}}\zeta'\mathbf{1}_{\gamma_{i_{1}}}(\zeta^{\frac{s_{2}}{d_{i_{2}}}}\zeta'\mathbf{1}_{\gamma_{i_{2}}})_{-}\right.\nn\\
		&\qquad \left.  +\delta_{i_{2},i_{3}}\zeta^{\frac{s_{2}+s_{3}}{d_{i_{2}}}}\zeta'\mathbf{1}_{\gamma_{i_{2}}}(\zeta^{\frac{s_{1}}{d_{i_{1}}}}\zeta'\mathbf{1}_{\gamma_{i_{1}}})_{-}- \delta_{i_{1},i_{2}}\zeta^{\frac{s_{1}+s_{2}}{d_{i_{1}}}}\zeta'\mathbf{1}_{\gamma_{i_{1}}}	
		(\zeta^{\frac{s_{3}}{d_{i_{3}}}}\zeta'\mathbf{1}_{\gamma_{i_{3}}})_{+}\right)dz\\
		=&	\frac{1}{2\pi \mathrm{i}}\oint_{\gamma_{i_{1}}}\left(-\delta_{i_{1},i_{2}}\delta_{i_{2},i_{3}}\zeta^{\frac{s_{1}+s_{2}+s_{3}}{d_{i_{1}}}}\zeta'\mathbf{1}_{\gamma_{i_{1}}}(\zeta'+\zeta'_{-}+\ell')dz\right)\\
		& +	\frac{1}{2\pi \mathrm{i}}\dsum_{s=1}^{m}\oint_{\gamma_{s}}\left(\delta_{i_{1},i_{2}}\zeta^{\frac{s_{1}
				+s_{2}}{d_{i_{1}}}}\zeta'\mathbf{1}_{\gamma_{i_{1}}}(\zeta^{\frac{s_{3}}{d_{i_{3}}}}\zeta'\mathbf{1}_{i_{3}})_{-}+\mathrm{c.p.}\{(i_{1},s_{1}),(i_{2},s_{2}),(i_{3},s_{3})\}\right)d z.
 	\end{align*}
	The other cases are similar. Thus the proposition is proved.
\end{proof}

\subsection{Potential}
Let us proceed to construct a function \( \mathcal{F} \) on the manifold \( \mathcal{M} \) such that
\begin{equation}\label{potentcon}
	\frac{\partial \mathcal{F}}{\partial u \partial v \partial w} = c\left(\frac{\partial}{\partial u}, \frac{\partial}{\partial v}, \frac{\partial}{\partial w}\right), \quad u, v, w \in \mathbf{u}.
\end{equation}

As a preparation, we have the following lemma.
\begin{lem}
There locally exist functions \( U_{i}(\mathbf{t}_{i}) \), \( V_{i}(\mathbf{h}_{0}) \), and \( W(\mathbf{h}) \)
depending analytically on the variables
\[
\mathbf{t}_{i} = \{t_{i,s} \mid s \in \mathbb{Z}\}, \quad \mathbf{h}_{0} = \{h_{0,j} \mid j=1,\ldots,n_0-1\},
\]
and $\mathbf{h}$ given in \eqref{flath}, respectively, such that
	\begin{align}
		&\frac{\partial^2 U_{i}}{\partial t_{i,s_{1}} \partial t_{i,s_{2}}} = \frac{1}{2\pi \mathrm{i}} \oint_{\gamma_{i}} \frac{\zeta'(z)\zeta(z)^{\frac{s_1+s_2}{d_{i}}}}{z-p_{i}} dz, 
\label{V1}
\\
		&\frac{\partial^2 V_{i}}{\partial h_{0,j_1} \partial h_{0,j_2}} = -\res_{z=\infty} \frac{\ell'(z)\chi_{0}(z)^{-j_1-j_2}}{z-p_{i}} dz, 
\label{V2} \\
		&\frac{\partial^3 W}{\partial u \partial v \partial w} = -\left(\res_{z=\infty} + \sum_{s=1}^{m} \res_{z=\varphi_{s}}\right) \frac{\partial_{\mathit{u}}\ell(z) \cdot \partial_{\mathit{v}}\ell(z) \cdot \partial_{\mathit{w}}\ell(z)}{\ell'(z)} dz, \quad u,v,w \in \mathbf{h}, \label{Suvw}
	\end{align}
where \( p_{i} \) denotes the center of the disk \( D_{i} \) for $i=1,\cdots,m$.
\end{lem}
\begin{proof}
It is easy to see:
	\begin{align*}
		&\frac{1}{2\pi \mathrm{i}} \frac{\partial}{\partial t_{i,s_{3}}} \oint_{\gamma_{i}} \frac{\zeta'\zeta^{\frac{s_1+s_2}{d_{i}}}}{z-p_{i}} dz \\
		&= \frac{1}{2\pi \mathrm{i}} \oint_{\gamma_{i}} \frac{-\left( \zeta'\zeta^{\frac{s_3}{d_{i}}} \right)' \cdot \zeta^{\frac{s_1+s_2}{d_{i}}} - \zeta' \cdot (\frac{s_1+s_2}{d_{i}})\zeta'\zeta^{\frac{s_1+s_2+s_{3}}{d_{i}}-1}}{z-p_{i}} dz \\
		&= \frac{1}{2\pi\mathrm{i}} \oint_{\gamma_{i}} \frac{-\left( \zeta'\zeta^{\frac{s_1+s_2+s_3}{d_{i}}} \right)'}{z-p_{i}} dz;
\\
		&-\frac{\partial}{\partial h_{0,j_3}} \text{Res}_{z=\infty} \frac{\ell'\chi_{0}^{-j_1-j_2}}{z-p_{i}} dz \\
		&= -\text{Res}_{z=\infty} \left( \left(\ell'\chi_{0}^{-j_3} \right)'_{\infty,\,\ge 0} \chi_{0}^{-j_1-j_2} - \ell' \cdot \frac{j_1+j_2}{n_{0}} \chi_{0}^{-j_1-j_2-n_{0}} \left(\ell'\chi_{0}^{-j_3} \right)_{\infty,\,\ge 0} \right) \frac{dz}{z-p_{i}} \\
		&= -\text{Res}_{z=\infty} \left( \left(\ell'\chi_{0}^{-j_3} \right)' \chi_{0}^{-j_1-j_2} + \left(\chi_{0}^{-j_1-j_2} \right)' \ell'\chi_{0}^{-j_3} \right) \frac{dz}{z-p_{i}} \\
	&= -\text{Res}_{z=\infty} \left(\ell'\chi_{0}^{-j_1-j_2-j_3} \right)' \frac{dz}{z-p_{i}},
	\end{align*}
which are symmetric with respect to the indices \( \{s_1, s_2 ,s_3\} \) and \( \{j_1 , j_2 , j_3 \} \) respectively. Hence the existence of \( U_{i}(\mathbf{t}_{i}) \) and \( V_{i}(\mathbf{h}_{0}) \) is verified.
It is known that the function \( W(\mathbf{h}) \) can be identified as the potential for the Frobenius manifold with superpotentials $\ell(z)$ (see for example \cite{dub1998,aoyama1996topological} and Subsection~\ref{McKP} below). The lemma is  proved.
\end{proof}

\begin{thm}\label{potent}
Let  $\mathcal{F}$ be a function on $\mathcal{M}$ given by
	\begin{align}
		\mathcal{F}=&\frac{1}{(2\pi \mathrm{i})^{2}}\dsum_{i=1}^{m}\oint_{\gamma_{i}}\oint_{\gamma_{i}^{\ast}}\left(\frac{1}{2}\zeta(z_{1}) \zeta(z_{2})
		+\zeta(z_{1}) \ell(z_{2}) -\ell(z_{1}) \zeta(z_{2})\right) \log\left(\frac{z_{1}-z_{2}}{z_{1}-p_{i}}\right)dz_{1}dz_{2}
\nn\\
		&+\frac{1}{(2\pi \mathrm{i})^{2}}\dsum_{i=1}^{m}\dsum_{i'\ne i}\oint_{\gamma_{i}}\oint_{\gamma_{i'}}\left(\frac{1}{2}\zeta(z_{1}) \zeta(z_{2})
		+\zeta(z_{1}) \ell(z_{2}) -\ell(z_{1}) \zeta(z_{2})\right) \log(z_{1}-z_{2})dz_{1}dz_{2}
\nn\\
		&\qquad -\dsum_{i=1}^{m}\frac{U_{i}(\mathbf{t}_{i})}{2\pi \mathrm{i} }\oint_{\gamma_{i}} \left(\frac{1}{2}\zeta(z)+\ell(z)\right) d z
		+ \dsum_{i=1}^{m}\frac{V_{i}(\mathbf{h}_{i})}{2\pi \mathrm{i} }\oint_{\gamma_{i}} \zeta(z) d z +W(\mathbf{h}). \label{potential}
	\end{align}
Then this function satisfies the condition \eqref{potentcon}.
	Here the contours \( \gamma_{i}^* \) are small deformations of \( \gamma_{i} \), satisfying $|z_{2}-p_{i}|<|z_{1}-p_{i}|$ for all $z_{2}\in\gamma_{i}$ and $z_{1}\in \gamma_{i}^{\ast}$ with \( i \in\{1,2, \ldots, m\} \).
\end{thm}
\begin{proof}
Let us verify the equalities \eqref{potentcon} in a straightforward way. Here in this proof the indices $i,i',i_\nu$ take value in $\{1,2,\dots,m\}$, and $s_\nu$ take value in $\mathbb{Z}$.

Firstly, we have the following equalities:
	\begin{align*}
	\omega(z_{1})_{-}=&\frac{1}{2\pi \mathrm{i}}\oint_{\gamma_{i}}\omega'(z_{2})\log(\frac{z_{1}-z_{2}}{z_{1}-p_{i}})dz_{2}+\frac{1}{2\pi \mathrm{i}}\sum_{i'\ne i}\oint_{\gamma_{i'}}\omega'(z_{2})\log(z_{1}-z_{2})dz_{2},
\\
\omega(z_{2})_{+}=&\frac{1}{2\pi \mathrm{i}}\oint_{\gamma_{i}}\frac{\omega(z)}{z-p_{i}}d z
-\frac{1}{2\pi \mathrm{i}}\oint_{\gamma_{i}^{\ast}}\omega'(z_{1})\log(\frac{z_{1}-z_{2}}{z_{1}-p_{i}})dz_{1} \\
	&-	\frac{1}{2\pi \mathrm{i}}\sum_{i'\ne i}\oint_{\gamma_{i'}}\omega'(z_{1})\log(z_{1}-z_{2})dz_{1}.
\end{align*}
Denote $Z_{i,s}(z)=\zeta(z)^{\frac{s}{d_{i}}}\zeta'(z)\mathbf{1}_{\gamma_{i}}$, then from Corollary \ref{flatvect} it follows that
\begin{align*}
	\frac{\p\zeta(z)}{\p t_{i,s}}=&-Z_{i,s}(z), \\ 
	 \frac{\p^2\zeta(z)}{\p t_{i_1,s_{1}}\p t_{i_2,s_{2}}}=&{\delta_{i_{1},i_{2}}Z_{i_1,s_{1}+s_{2}}}'(z), \\
	 \frac{\p^3\zeta(z)}{\p t_{i_1,s_{1}}\p t_{i_2,s_{2}}\p t_{i_3,s_{3}}}=&-{\delta_{i_{1},i_{2}}\delta_{i_{2},i_{3}}Z_{i_1,s_{1}+s_2+s_3}}''(z).
\end{align*}
By using such equalities, we have
{\footnotesize\begin{align*}
	&\frac{\p^{3} \mathcal{F}}{\p t_{i_{1},s_{1}} \p t_{i_{2},s_{2}}\p t_{i_{3},s_{3}}} \\
	=&\frac{1}{(2\pi \mathrm{i})^{2}}\dsum_{i=1}^{m}\oint_{\gamma_{i}}\oint_{\gamma_{i}^{\ast}}\bigg( -\delta_{i_{1},i_{2}}\delta_{i_{2},i_{3}}\left(Z_{i_{1},s_1+s_2+s_3}''(z_1)\left(\frac{1}{2}\zeta(z_2)+\ell(z_2)\right)\right. \\
&+\left. \left(\frac{1}{2}\zeta(z_1)-\ell(z_1)\right){Z_{i_{1},s_1+s_2+s_3}}''(z_2) \right)
\\
	&  -\frac{1}{2}\Big(\delta_{i_{1},i_{2}}(Z_{i_{1},s_1+s_2}'(z_1)Z_{i_{3},s_3}(z_2) +Z_{i_{1},s_1+s_2}'(z_2)Z_{i_{3},s_3}(z_1)) \\
&+\mathrm{c.p.}\{(i_{1},s_{1}),(i_{2},s_{2}),(i_{3},s_{3})\}\Big) \bigg)
	\log\left(\frac{z_{1}-z_{2}}{z_{1}-p_{i}}\right)  dz_{1}dz_{2}
\\
		&+\frac{1}{(2\pi \mathrm{i})^{2}}\dsum_{i=1}^{m}\dsum_{s'\ne s}\oint_{\gamma_{i}}\oint_{\gamma_{s'}}\bigg( -\delta_{i_{1},i_{2}}\delta_{i_{2},i_{3}}\left(Z_{i_{1},s_1+s_2+s_3}''(z_1)\left(\frac{1}{2}\zeta(z_2)+\ell(z_2)\right)\right.
\\
&+ \left(\frac{1}{2}\zeta(z_1)-\ell(z_1)\right){Z_{i_{1},s_1+s_2+s_3}}''(z_2) \bigg)
\\
	& -\frac{1}{2}\Big(\delta_{i_{1},i_{2}}(Z_{i_{1},s_1+s_2}'(z_1)Z_{i_{3},s_3}(z_2) +Z_{i_{1},s_1+s_2}'(z_2)Z_{i_{3},s_3}(z_1))
\\
&+\mathrm{c.p.}\{(i_{1},s_{1}),(i_{2},s_{2}),(i_{3},s_{3})\}\Big) \bigg)
	\log(z_{1}-z_{2}) dz_{1}dz_{2}\\
	&-\frac{1}{2\pi \mathrm{i} }\dsum_{i=1}^{m}\oint_{\gamma_{i}} \left(\frac{1}{2}\zeta(z)
	+\ell(z)\right) d z \cdot \frac{1}{2\pi \mathrm{i}}\oint_{\gamma_{i}}\frac{-\delta_{i_{1},i_{2}}\delta_{i_{2},i_{3}}Z_{i_{1},s_1+s_2+s_3}'(z)}{z-p_{i}}d z \\
	&-\left\{\frac{1}{2\pi \mathrm{i} }\dsum_{i=1}^{m}\oint_{\gamma_{i}} \left(-\frac{1}{2}Z_{i_3,s_3}(z)\right) d z \cdot \frac{1}{2\pi \mathrm{i}}
	\oint_{\gamma_{i}}\frac{\delta_{i_{1},i_{2}}Z_{i_{1},s_1+s_2}(z)}{z-p_{i}}d z +\mathrm{c.p.}\{(i_{1},s_{1}),(i_{2},s_{2}),(i_{3},s_{3})\}\right\}
	\\
	=&\frac{1}{2\pi \mathrm{i}}\sum_{i=1}^{m}	\oint_{\gamma_{i}}\left(  \delta_{i_{1},i_{2}}\delta_{i_{2},i_{3}}\left((Z_{i_{1},s_1+s_2+s_3}'(z))_{+} \left(\frac{1}{2}\zeta(z)+\ell(z)\right)
	- \left(\frac{1}{2}\zeta(z)-\ell(z)\right)(Z_{i_{1},s_1+s_2+s_3}'(z))_{-}\right) \right. \\
	& -\frac{1}{2}(-\delta_{i_{1},i_{2}}(Z_{i_{1},s_1+s_2}(z))_{+}Z_{i_{3},s_{3}}(z) +\delta_{i_{1},i_{2}}Z_{i_3,s_{3}}(z)(Z_{i_{1},s_1+s_2}(z))_{-}
\\
&+\mathrm{c.p.}\{(i_{1},s_{1}),(i_{2},s_{2}),(i_{3},s_{3})\})\bigg)
	d z \\
	&-\frac{1}{2\pi \mathrm{i} }\dsum_{i=1}^{m}\oint_{\gamma_{i}} \left(\frac{1}{2}\zeta(z)
+\ell(z)\right) d z \cdot \frac{1}{2\pi \mathrm{i}}\oint_{\gamma_{i}}\frac{\delta_{i_{1},i_{2}}\delta_{i_{2},i_{3}}Z_{i_{1},s_1+s_2+s_3}'(z)}{z-p_{i}}d z \\
&+\left\{\frac{1}{2\pi \mathrm{i} }\dsum_{i=1}^{m}\oint_{\gamma_{i}} \left(-\frac{1}{2}Z_{i_3,s_3}(z)\right) d z \cdot \frac{1}{2\pi \mathrm{i}}
\oint_{\gamma_{i}}\frac{\delta_{i_{1},i_{2}}Z_{i_{1},s_1+s_2}(z)}{z-p_{i}}d z +\mathrm{c.p.}\{(i_{1},s_{1}),(i_{2},s_{2}),(i_{3},s_{3})\}\right\}
\\
		&-\frac{1}{2\pi \mathrm{i} }\dsum_{i=1}^{m}\oint_{\gamma_{i}} \left(\frac{1}{2}\zeta(z)
	+\ell(z)\right) d z \cdot \frac{1}{2\pi \mathrm{i}}\oint_{\gamma_{i}}\frac{-\delta_{i_{1},i_{2}}\delta_{i_{2},i_{3}}Z_{i_{1},s_1+s_2+s_3}'(z)}{z-p_{i}}d z \\
	&-\left\{\frac{1}{2\pi \mathrm{i} }\dsum_{i=1}^{m}\oint_{\gamma_{i}} \left(-\frac{1}{2}Z_{i_3,s_3}(z)\right) d z \cdot \frac{1}{2\pi \mathrm{i}}
	\oint_{\gamma_{i}}\frac{\delta_{i_{1},i_{2}}Z_{i_{1},s_1+s_2}(z)}{z-p_{i}}d z +\mathrm{c.p.}\{(i_{1},s_{1}),(i_{2},s_{2}),(i_{3},s_{3})\}\right\}
	\\
	=&\frac{1}{2\pi \mathrm{i}}\dsum_{i=1}^{m}\oint_{\gamma_{i}} \left(  \delta_{i_{1},i_{2}}\delta_{i_{2},i_{3}}Z_{i_{1},s_1+s_2+s_3}(z) \left( -\frac{1}{2}\zeta'(z)_- -\ell'(z)_-
	+ \frac{1}{2}\zeta'(z)_+ -\ell'(z)_+ \right)  \right. \\
	&\left. - \frac{3}{2} Z_{i_3,s_3}(z)\delta_{i_{1},i_{2}}Z_{i_{1},s_1+s_2}(z) + (\delta_{i_{1},i_{2}}(Z_{i_{1},s_1+s_2}(z))_{+}Z_{i_3,s_3}(z)  +\mathrm{c.p.}\{(i_{1},s_{1}),(i_{2},s_{2}),(i_{3},s_{3})\}) \right)
	d z
	\\
	=&\frac{1}{2\pi \mathrm{i}}\dsum_{i=1}^{m}\oint_{\gamma_{i}} \left(  \delta_{i_{1},i_{2}}\delta_{i_{2},i_{3}}Z_{i_{1},s_1+s_2+s_3}(z) \left(-\zeta'(z) - \zeta'(z)_- -\ell'(z) \right)\right.\\
	&\left.
	+ (\delta_{i_{1},i_{2}}Z_{i_{1},s_1+s_2}(z) (Z_{i_3,s_3}(z))_{-} +\mathrm{c.p.}\{(i_{1},s_{1}),(i_{2},s_{2}),(i_{3},s_{3})\} \right)
	d z \\
	=& \left\langle \frac{\p}{\p t_{i_{1},s_{1}}}\circ \frac{\p}{\p t_{i_{2},s_{2}}},
	\frac{\p}{\p t_{i_{3},s_{3}}}\right\rangle_{\eta}.
\end{align*}}
The other cases are similar. Therefore the theorem is proved.
\end{proof}

\subsection{Euler vector field}\label{euler}

Recall the definition of flat coordinates \eqref{flatt} and \eqref{flath}. Let us consider a vector field on \( \mathcal{M} \) as
\begin{align}
E = & \sum_{i=1}^{m} \sum_{s \in \mathbb{Z}} \left( \frac{1}{n_{0}} - \frac{s}{d_{i}} \right) t_{i,s} \frac{\partial}{\partial t_{i,s}} + \sum_{j=1}^{n_{0}-1}   \frac{1 + j}{n_{0}}   h_{0,j} \frac{\partial}{\partial h_{0,j}}
\nn \\
& + \sum_{k=1}^{m} \sum_{r=0}^{n_{k}} \left( \frac{1}{n_{0}} + \frac{r}{n_{k}} \right) h_{k,r} \frac{\partial}{\partial h_{k,r}}. \label{E}
\end{align}
By using Corollary~\ref{flatvect}, we have
$$
\p_{E}\vec{a}=\left(a(z)-\frac{z}{n_{0}}a'(z),\hat{a}(z)-\frac{z}{n_{0}}\hat{a}'(z)\right).
$$
\begin{prop}\label{hom}
	The function $\mathcal{F}$ given in \eqref{potential} satisfies
	\[
	\operatorname{Lie}_{E} \mathcal{F} = \left(2 + \frac{2}{n_{0}}\right) \mathcal{F} + \text{quadratic terms in } \mathbf{u}.
	\]
\end{prop}
\begin{proof}
%
Let us assign $\deg\,z=\frac{1}{n_{0}}$, and assume each of the functions $\zeta(z)$ and $\ell(z)$ to be homogeneous of degree $1$. It is easy to see that the flat coordinates of $\mathcal{M}$ have degrees:
\begin{align}\label{degree}
	&\deg\,t_{i,s}=\frac{1}{n_{0}}-\frac{s}{d_{i}},\quad i\in\Z; \\
	&  \deg\,h_{0,j}=\frac{j+1}{n_{0}},\quad
	1\le j\le n_{0}-1; \\
	&     \deg\,h_{k,r}=\frac{r}{n_{k}}+\frac{1}{n_{0}}, \quad 0\le r\le n_{k}. \label{degree3}
\end{align}
Accordingly, one sees that the functions $U_{i}$ and $V_{i}$ defined by \eqref{V1} and \eqref{V2} are homogeneous of degree $1+\frac{1}{n_{0}}$ (up to a linear terms in the flat coordinates),
while $W$ given by \eqref{Suvw} is homogeneous of degree $2+\frac{2}{n_{0}}$ (up to a quadratic terms in the flat coordinates).
Thus the function $\mathcal{F}$ is homogeneous of degree $2+\frac{2}{n_{0}}$ up to quadratic terms.
Therefore the proposition is proved.
\end{proof}
%
%
%
%
%

\subsection{Intersection form}
In analogue to the case of finite dimension, now we introduce an intersection form  on $\mathcal{M}$, which will be applied to construct the principal hierarchy in the next section.

Recall the map $C_{\vec\xi}$ defined by \eqref{whiprod}. According to Corollary~\ref{Cxieta}, with the help of the Euler vector field, it is well defined a linear map $g$ from $T_{\vec{a}}^\star\mathcal{M}$ to $T_{\vec{a}}\mathcal{M}$ as
\begin{equation}\label{gmap}
g(\vec{\omega})=C_{E}(\vec{\omega}),\quad \vec{\omega}\in T_{\vec{a}}^\star\mathcal{M}.
\end{equation}
Let us introduce a bilinear form on the cotangent space $T_{\vec{a}}^\star\mathcal{M}$ as
\begin{equation}\label{blfg}
  (\vec{\omega}_{1},\vec{\omega}_{2})_{g}=\langle \vec{\omega}_{1},g(\vec{\omega}_{2})\rangle,\quad \vec{\omega}_{1},\vec{\omega}_{2}\in T_{\vec{a}}^\star\mathcal{M}.
\end{equation}
We call this bilinear form the intersection form on $\mathcal{M}$.
\begin{prop}
The intersection form \eqref{blfg} is symmetric. More explicitly,
for $\vec{\omega}_{\nu}=(\omega_\nu(z),\hat{\omega}_\nu(z))\in T_{\vec{a}}^\star\mathcal{M}$ with $\nu\in\{1,2\}$, it holds that
\begin{align} (\vec{\omega}_{1},\vec{\omega}_{2})_{g}=&\frac{1}{2\pi\mathrm{i}}\sum_{j=1}^{m}\oint_{\gamma_{j}}((\omega_{1}a'+\hat{\omega}_{1}\hat{a}')_{+}(a\omega_{2}+\hat{a}\hat{\omega}_{2})_{-}+(\omega_{2}a'+\hat{\omega}_{2}\hat{a}')_{+}(a\omega_{1}+\hat{a}\hat{\omega}_{1})_{-}) dz
\nn\\	&-\frac{1}{2\pi\mathrm{i}}\sum_{j=1}^{m}\oint_{\gamma_{j}}(\omega_{1}\omega_{2}aa'+\omega_{1}\hat{\omega}_{2}a\hat{a}' +\omega_{2}\hat{\omega}_{1}a\hat{a}'+\hat{\omega}_{1}\hat{\omega}_{2}\hat{a}\hat{a}')dz
\nn\\
&+\frac{1}{n_{0}}\frac{1}{(2\pi\mathrm{i})^2}\left(\sum_{j=1}^{m}\oint_{\gamma_{j}}(a'\omega_{1} +\hat{a}'\hat{\omega}_{1})dz\right)\left(\sum_{j=1}^{m}\oint_{\gamma_{j}}(a'\omega_{2}+\hat{a}'\hat{\omega}_{2})dz\right). \label{interom12}
\end{align}
\end{prop}
\begin{proof}
Firstly, by using \eqref{eq-Cxi} we have
$$
(\vec{\omega}_{1},\vec{\omega}_{2})_{g}=\langle\vec{\omega}_{1}\circ \vec{\omega}_{2}
,E\rangle.
$$
Hence the intersection form is symmetric. By using \eqref{whiprod}, one has
\begin{align*}
	g(\vec{\omega})=&(a'(\omega a+\hat{\omega}\hat{a})_{-}-a(\omega a'+\hat{\omega}\hat{a}')_{-}+\frac{a'}{n_{0}}\frac{1}{2\pi\mathrm{i}}\sum_{j=1}^{m}\oint_{\gamma_{j}}(\omega a'+\hat{\omega}\hat{a}')dz,\\
	&-\hat{a}'(\omega a+\hat{\omega}\hat{a})_{+}+\hat{a}(\omega a'+\hat{\omega}\hat{a}')_{+}+\frac{\hat{a}'}{n_{0}}\frac{1}{2\pi\mathrm{i}}\sum_{j=1}^{m}\oint_{\gamma_{j}}(\omega a'+\hat{\omega}\hat{a}')dz)
\end{align*}
 with the help of
$$
z\omega_{-}-(z\omega)_{-}=\frac{1}{2\pi\mathrm{i}}\sum_{j=1}^{m}\oint_{\gamma_{j}}\omega dz,\quad \omega\in\mathcal{H}.
$$
Then, the equality \eqref{interom12} follows from a straightforward calculation. The proposition is proved.
\end{proof}

\section{Principal hierarchy for $\mathcal{M}$}

In this section, we proceed to construct the principal hierarchy for the infinite-dimensional Frobenius manifold $\mathcal{M}$, and show that one of its subhierarchy coincides with the genus-zero Whitham hierarchy. Inspired by reduction properties of the Whitham hierarchy, 
we will study the relationship between $\mathcal{M}$ and some finite-dimensional Frobenius manifolds with rational superpotentials.

\subsection{General theory of principal hierarchy}
As a preparation, we review the general theory of the principal hierarchy \cite{dub1998} for a Frobenius manifold $(M,\eta,\circ,e,E)$, where $\eta$ is the flat metric, $\circ$ is the product of vector fields, $e$ and $E$ are the unit vector field and the Euler vector field respectively. Let $\nabla$ be the Levi-Civita connection of the metric $\eta$.
The deformed flat connection on \(M\), originally introduced by Dubrovin \cite{dub1998}, is given by
\[
\tilde{\nabla}_X Y = \nabla_X Y + z X \circ Y, \quad X, Y \in \operatorname{Vect}(M).
\]
This connection extends to a flat connection on \(M \times \mathbb{C}^*\) such that:
\[
\begin{aligned}
	&\tilde{\nabla}_X \frac{d}{d z} = 0, \\
	&\tilde{\nabla}_{\frac{d}{d z}} \frac{d}{d z} = 0, \\
	&\tilde{\nabla}_{\frac{d}{d z}} X = \partial_z X + E \circ X - \frac{1}{z} \mathcal{V}(X).
\end{aligned}
\]
Here \(X\) is a vector field on \(M \times \mathbb{C}^*\) that is trivial along the \(\mathbb{C}^*\) component, and ($d$ is the charge of $M$)
\[
\mathcal{V}(X) := \frac{2-d}{2} X - \nabla_X E.
\]

One chooses a system of flat coordinates $t=(t^{1},\cdots,t^{n})$ with respect to the metric $\eta$, such that
\[
\mathcal{V}\left(\frac{\partial}{\partial t^\alpha}\right)=\mu_\alpha \frac{\partial}{\partial t^\alpha}, \quad \alpha = 1, \ldots, n,
\]
with \( \mu_1, \ldots, \mu_n \) being the spectrum of $M$.
Suppose that the linear differential equation  \( \widetilde{\nabla} \sigma = 0 \) of an unknown $1$-form $\sigma$ has the following fundamental solutions:
\[
\sigma_\alpha(t,z) = \frac{\partial \tilde{v}_\alpha(t,z)}{\partial t^\beta} d t^\beta ~~ \hbox{with} ~~ \alpha \in\{ 1, \ldots, n\}, \quad  \sigma_{n+1}(t,z) = d z.
\]
Here the functions \(\tilde{v}_1(t, z), \ldots, \tilde{v}_n(t, z)\) can be normalized as
\[
\left(\tilde{v}_1(t, z), \ldots, \tilde{v}_n(t, z)\right) = \left(\theta_1(t, z), \ldots, \theta_n(t, z)\right) z^\mu z^R,
\]
where \( \theta_\alpha(t, z) \) are analytic functions in a neighborhood of \( z = 0 \) such that
\begin{equation}\label{princon1}
	\theta_{\alpha}(t,0) = \eta_{\alpha \beta} t^\beta,
\end{equation}
 and the constant matrices
 \[  \mu = \operatorname{diag}(\mu_1, \ldots, \mu_n), \quad  R = R_1 + \ldots + R_n
\]
 with $R$ being nilpotent and satisfying
\begin{align*}
	& (R_s)_\beta^\alpha = 0 ~~\text{ if }~~ \mu_\alpha - \mu_\beta \neq s, \\
	& (R_s)_\alpha^\gamma \eta_{\gamma \beta} = (-1)^{s+1} (R_s)_\beta^\gamma \eta_{\gamma \alpha}.
\end{align*}
Here and below, a matrix $R=(R_\beta^\alpha)$ is said to be nilpotent if there is a positive integer $N$ such that the product $R_{\alpha_2}^{\alpha_1}R_{\alpha_3}^{\alpha_2}\dots R_{\alpha_{N}}^{\alpha_{N-1}}$ vanishes
 for any indices $\alpha_1,\alpha_2, \dots, \alpha_N$.

Let $c_{\beta \gamma}^\epsilon(t)$ be the structure coefficients for the basis $\{\partial/\partial t^\alpha\}_{\alpha=1}^{n}$. The coefficients in the expansions (about $z=0$)
\[
\theta_\alpha(t, z) = \sum_{p \geq 0} \theta_{\alpha, p}(t) z^p, \quad \alpha = 1, \ldots, n,
\]
satisfy that
\begin{align}
\label{princon2}
 \theta_{\alpha,0}=& \eta_{\alpha \beta} t^\beta, \\
\label{princon3}
	\frac{\partial^2 \theta_{\alpha, p+1}(t)}{\partial t^\beta \partial t^\gamma} =& c_{\beta \gamma}^\epsilon(t) \frac{\partial \theta_{\alpha, p}(t)}{\partial t^\epsilon},
\\
	\operatorname{Lie}_E\left(\frac{\partial \theta_{\alpha, p}(t)}{\partial t^\beta}\right) =& \left(p + \mu_\alpha + \mu_\beta\right) \frac{\partial \theta_{\alpha, p}(t)}{\partial t^\beta} +\sum_{s=1}^p  \left(R_{s}\right)_\alpha^\epsilon\frac{\partial\theta_{\epsilon, p-s}(t)}{\partial t^\beta}. \label{princon4}
\end{align}

Given a system of functions \(\left\{\theta_{\alpha, p}\right\}\) satisfying the conditions \eqref{princon4}--\eqref{princon2},
the principal hierarchy for \(M\) means the following system of Hamiltonian equations on the loop space \(\left\{S^1 \rightarrow M\right\}\):
\begin{equation}\label{ham1}
\frac{\partial t^\gamma}{\partial T^{\alpha, p}} = \left\{t^\gamma(x), \int \theta_{\alpha, p+1} (t) \, d x\right\}_1 := \eta^{\gamma \beta} \frac{\partial}{\partial x}\left(\frac{\partial \theta_{\alpha, p+1} (t)}{\partial t^\beta}\right), \end{equation}
where $\alpha, \beta \in\{ 1,2, \ldots, n\}$ and $\ p \geq 0$.
In generic case, these flows can be written in the bi-Hamiltonian recursion form via
\begin{align*}
\left\{t^\gamma(x), \int \theta_{\alpha, p} (t) \, d x\right\}_2 =&\left(p + \mu_\alpha + \frac{1}{2}\right) \left\{t^\gamma(x), \int \theta_{\alpha, p+1} (t) \, d x\right\}_1 \\
 &+ \sum_{s=1}^p \left(R_{s}\right)_\alpha^\epsilon \left\{t^\gamma(x), \int \theta_{\epsilon, p-s+1} (t) \, d x\right\}_1.
\end{align*}
Note that every linear combination of $\{\ ,\ \}_1$ and $\{\ ,\ \}_2$ is still a Hamiltonian structure, namely, they compose a bi-Hamiltonian structure.

To summarize, in order to construct the principal hierarchy for the infinite-dimensional Frobenius manifold $\mathcal{M}$ constructed in the previous section, it suffices to calculate the Hamiltonian structure as in \eqref{ham1}, as well as a set of Hamiltonian densities satisfying the conditions like \eqref{princon2}--\eqref{princon4}.

\subsection{Bi-Hamiltonian structure on the loop space of $\mathcal{M}$}

As to be seen, on the loop space $\{S^1\to\mathcal{M}\}$ there are two compatible Hamiltonian structures, which correspond to the flat metric $\langle\ ,\ \rangle_{\eta}$ in \eqref{whimetric} and the intersection form $\langle\ ,\ \rangle_{g}$ in \eqref{blfg} respectively.

\begin{lem}\label{nabla}
For any two vector fields \(\partial_1\) and \(\partial_2\) on \(\mathcal{M}\), 
the covariant derivative of \(\partial_2\) along \(\partial_1\) with respect to the metric $\langle\ ,\ \rangle_{\eta}$ is given by
		\begin{align}
			\left(\nabla_{\partial_1} \partial_2\right) (\vec{a})= & \bigg(\partial_1 \partial_2 a(z)-\left(\frac{\partial_1 \zeta(z) \partial_2 \zeta(z)}{\zeta^{\prime}(z)}\right)_{-}^{\prime} -\left(\frac{\partial_1 \ell(z) \partial_2 \ell(z)}{\ell^{\prime}(z)}\right)_{\infty,\,\ge 0}^{\prime} \nn\\
&-\sum_{j=1}^{m}\left(\frac{\partial_1 \ell(z) \partial_2 \ell(z)}{\ell^{\prime}(z)}\right)_{\varphi_{j},\,\le -1}^{\prime},
			 \   \partial_1 \partial_2 \hat{a}(z)+\left(\frac{\partial_1 \zeta(z) \partial_2 \zeta(z)}{\zeta^{\prime}(z)}\right)_{+}^{\prime}
\nn\\
&-\left(\frac{\partial_1 \ell(z) \partial_2 \ell(z)}{\ell^{\prime}(z)}\right)_{\infty,\,\ge 0}^{\prime}-\sum_{j=1}^{m}\left(\frac{\partial_1 \ell(z) \partial_2 \ell(z)}{\ell^{\prime}(z)}\right)_{\varphi_{j},\,\le -1}^{\prime}\bigg) . \label{conn}
		\end{align}
\end{lem}
\begin{proof}
It suffices to verify that the connection $\nabla$ defined by equality \eqref{conn} is torsion-free and compatible with $\langle\ ,\ \rangle_\eta$.

Firstly, from \eqref{conn} it is easy to check 
	$$
	\left(\nabla_{\partial_1} \partial_2-\nabla_{\partial_2} \partial_1\right) ( \vec{a})=\partial_1 \partial_2 \vec{a}-\partial_2 \partial_1 \vec{a},
	$$
which implies that the connection  $\nabla$ is torsion-free.
	
On the other hand, for any vector field $\partial_3$, by using  \eqref{whimetric} we have
		\begin{align}
			& {\partial_1}\left\langle\partial_2, \partial_3\right\rangle_\eta
\nn\\
			= & -\frac{1}{2 \pi \mathrm{i}} \dsum_{i=1}^{m}\oint_{\gamma_{i}}\left( \frac{\partial_1 \partial_2 \zeta(z) \cdot \partial_3 \zeta(z)}{\zeta^{\prime}(z)} + \frac{\partial_2 \zeta(z) \cdot \partial_1 \partial_3 \zeta(z)}{\zeta^{\prime}(z)} - \frac{\partial_2 \zeta(z) \cdot \partial_3 \zeta(z) \cdot \partial_1 \zeta^{\prime}(z)}{\left(\zeta^{\prime}(z)\right)^2}\right) d z
\nn\\
			& -\operatorname{Res}_{z=\infty} \left( \frac{\partial_1 \partial_2 \ell(z) \cdot \partial_3 \ell(z)}{\ell^{\prime}(z)} + \frac{\partial_2 \ell(z) \cdot \partial_1 \partial_3 \ell(z)}{\ell^{\prime}(z)} - \frac{\partial_2 \ell(z) \cdot \partial_3 \ell(z) \cdot \partial_1 \ell^{\prime}(z)}{\left(\ell^{\prime}(z)\right)^2}\right) d z
\nn\\
			& -\sum_{i=1}^{m}\operatorname{Res}_{z=\varphi_{i}}\left( \frac{\partial_1 \partial_2 \ell(z) \cdot \partial_3 \ell(z)}{\ell^{\prime}(z)}  + \frac{\partial_2 \ell(z) \cdot \partial_1 \partial_3 \ell(z)}{\ell^{\prime}(z)} -\frac{\partial_2 \ell(z) \cdot \partial_3 \ell(z) \cdot \partial_1 \ell^{\prime}(z)}{\left(\ell^{\prime}(z)\right)^2} \right)dz. \label{compfor1}
		\end{align}
In particular, by integration by parts, we obtain
	\begin{align}
			& \frac{1}{2 \pi \mathrm{i}} \dsum_{i=1}^{m}\oint_{\gamma_{i}} \frac{\partial_2 \zeta(z) \cdot \partial_3 \zeta(z) \cdot \partial_1 \zeta^{\prime}(z)}{\left(\zeta^{\prime}(z)\right)^2} d z
\nn\\
			= & -\frac{1}{2 \pi \mathrm{i}} \dsum_{i=1}^{m}\oint_{\gamma_{i}}\left(\frac{\partial_2 \zeta(z)}{\zeta^{\prime}(z)}\right)^{\prime} \frac{\partial_3 \zeta(z) \cdot \partial_1 \zeta(z)}{\zeta^{\prime}(z)} d z-\frac{1}{2 \pi \mathrm{i}} \dsum_{i=1}^{m}\oint_{\gamma_{i}}\left(\frac{\partial_3 \zeta(z)}{\zeta^{\prime}(z)}\right)^{\prime} \frac{\partial_2 \zeta(z) \cdot \partial_1 \zeta(z)}{\zeta^{\prime}(z)} d z
\nn\\
			= & \frac{1}{2 \pi \mathrm{i}} \dsum_{i=1}^{m}\oint_{\gamma_{i}} \frac{\partial_2 \zeta(z)}{\zeta^{\prime}(z)}\left(\frac{\partial_3 \zeta(z) \cdot \partial_1 \zeta(z)}{\zeta^{\prime}(z)}\right)^{\prime} d z+\frac{1}{2 \pi \mathrm{i}} \dsum_{i=1}^{m}\oint_{\gamma_{i}} \frac{\partial_3 \zeta(z)}{\zeta^{\prime}(z)}\left(\frac{\partial_2 \zeta(z) \cdot \partial_1 \zeta(z)}{\zeta^{\prime}(z)}\right)^{\prime} d z, \label{compfor2}
\\
			& \operatorname{Res}_{z=\infty} \frac{\partial_2 \ell(z) \cdot \partial_3 \ell(z) \cdot \partial_1 \ell^{\prime}(z)}{\left(\ell^{\prime}(z)\right)^2} d z
\nn\\
			= & \operatorname{Res}_{z=\infty} \frac{\partial_2 \ell(z)}{\ell^{\prime}(z)}\left(\frac{\partial_3 \ell(z) \cdot \partial_1 \ell(z)}{\ell^{\prime}(z)}\right)^{\prime} d z+\operatorname{Res}_{z=\infty} \frac{\partial_3 \ell(z)}{\ell^{\prime}(z)}\left(\frac{\partial_2 \ell(z) \cdot \partial_1 \ell(z)}{\ell^{\prime}(z)}\right)^{\prime} d z, \label{compfor3}
\\
			& \operatorname{Res}_{z=\varphi_{i}} \frac{\partial_2 \ell(z) \cdot \partial_3 \ell(z) \cdot \partial_1 \ell^{\prime}(z)}{\left(\ell^{\prime}(z)\right)^2} d z
\nn\\
			= & \operatorname{Res}_{z=\varphi_{i}} \frac{\partial_2 \ell(z)}{\ell^{\prime}(z)}\left(\frac{\partial_3 \ell(z) \cdot \partial_1 \ell(z)}{\ell^{\prime}(z)}\right)^{\prime} d z+\operatorname{Res}_{z=\varphi_{i}} \frac{\partial_3 \ell(z)}{\ell^{\prime}(z)}\left(\frac{\partial_2 \ell(z) \cdot \partial_1 \ell(z)}{\ell^{\prime}(z)}\right)^{\prime} d z . \label{compfor4}
		\end{align}
Substituting them into \eqref{compfor1}, we arrive at
		\begin{align}
			&  {\partial_1}\left\langle\partial_2, \partial_3\right\rangle_\eta
\nn\\
			=&-\frac{1}{2 \pi \mathrm{i}} \dsum_{i=1}^{m}\oint_{\gamma_{i}} \frac{\partial_1 \partial_2 \zeta(z) \cdot \partial_3 \zeta(z)}{\zeta^{\prime}(z)} d z-\frac{1}{2 \pi \mathrm{i}} \dsum_{i=1}^{m}\oint_{\gamma_{i}} \frac{\partial_2 \zeta(z) \cdot \partial_1 \partial_3 \zeta(z)}{\zeta^{\prime}(z)} d z
\nn\\
			&+\frac{1}{2 \pi \mathrm{i}} \dsum_{i=1}^{m}\oint_{\gamma_{i}} \frac{\partial_2 \zeta(z)}{\zeta^{\prime}(z)}\left(\frac{\partial_3 \zeta(z) \cdot \partial_1 \zeta(z)}{\zeta^{\prime}(z)}\right)^{\prime} d z
\nn\\
&+\frac{1}{2 \pi \mathrm{i}} \dsum_{i=1}^{m}\oint_{\gamma_{i}} \frac{\partial_3 \zeta(z)}{\zeta^{\prime}(z)}\left(\frac{\partial_2 \zeta(z) \cdot \partial_1 \zeta(z)}{\zeta^{\prime}(z)}\right)^{\prime} d z
\nn\\
			&-\operatorname{Res}_{z=\infty} \frac{\partial_1 \partial_2 \ell(z) \cdot \partial_3 \ell(z)}{\ell^{\prime}(z)} d z-\operatorname{Res}_{z=\infty} \frac{\partial_2 \ell(z) \cdot \partial_1 \partial_3 \ell(z)}{\ell^{\prime}(z)} d z
\nn\\
			&+\operatorname{Res}_{z=\infty} \frac{\partial_2 \ell(z)}{\ell^{\prime}(z)}\left(\frac{\partial_3 \ell(z) \cdot \partial_1 \ell(z)}{\ell^{\prime}(z)}\right)^{\prime} d z+\operatorname{Res}_{z=\infty} \frac{\partial_3 \ell(z)}{\ell^{\prime}(z)}\left(\frac{\partial_2 \ell(z) \cdot \partial_1 \ell(z)}{\ell^{\prime}(z)}\right)^{\prime} d z
\nn\\
			&-\sum_{i=1}^{m}\operatorname{Res}_{z=\varphi_{i}}\left(  \frac{\partial_1 \partial_2 \ell(z) \cdot \partial_3 \ell(z)}{\ell^{\prime}(z)} + \frac{\partial_2 \ell(z) \cdot \partial_1 \partial_3 \ell(z)}{\ell^{\prime}(z)}\right) d z
\nn\\
			&+\sum_{i=1}^{m}\operatorname{Res}_{z=\varphi_{i}} \frac{\partial_2 \ell(z)}{\ell^{\prime}(z)}\left(\frac{\partial_3 \ell(z) \cdot \partial_1 \ell(z)}{\ell^{\prime}(z)}\right)^{\prime} d z
\nn\\
&+\sum_{i=1}^{m}\operatorname{Res}_{z=\varphi_{i}} \frac{\partial_3 \ell(z)}{\ell^{\prime}(z)}\left(\frac{\partial_2 \ell(z) \cdot \partial_1 \ell(z)}{\ell^{\prime}(z)}\right)^{\prime} d z . \label{compfor5}
		\end{align}
Furthermore, from \eqref{conn} and \eqref{zetaell} it follows that
	\begin{align}
		&\left(\nabla_{\partial_1} \partial_2\right) ( \zeta(z))=\partial_1 \partial_2 \zeta(z)-\left(\frac{\partial_1 \zeta(z) \partial_2 \zeta(z)}{\zeta^{\prime}(z)}\right)^{\prime},\label{compfor6}\\
		&	\left(\nabla_{\partial_1} \partial_2\right) ( \ell(z) )=\partial_1 \partial_2 \ell(z)-\left(\frac{\partial_1 \ell(z) \partial_2 \ell(z)}{\ell^{\prime}(z)}\right)_{\infty,\,\ge 0}^{\prime}-\sum_{j=1}^{m}\left(\frac{\partial_1 \ell(z) \partial_2 \ell(z)}{\ell^{\prime}(z)}\right)_{\varphi_{j},\,\le -1}^{\prime} .\label{compfor7}
	\end{align}
Thus, taking \eqref{compfor5}--\eqref{compfor7} together, we achieve
	\begin{equation}\label{compa}
		 {\partial_1}\left\langle\partial_2, \partial_3\right\rangle_\eta=\left\langle\nabla_{\partial_1} \partial_2, \partial_3\right\rangle_\eta+\left\langle\partial_2, \nabla_{\partial_1} \partial_3\right\rangle_\eta,
	\end{equation}
which means that $\nabla$ is compatible with the metric $\langle \ ,\ \rangle_\eta$. Therefore the lemma is proved.
\end{proof}

For any covector field  $\vec{\omega}$ on $\mathcal{M}$, its covariant derivative is given by
\begin{equation}\label{nablaom}
\left\langle \nabla_{\partial_1} \vec{\omega}, \partial_2\right\rangle = \partial_1\left\langle\vec{\omega}, \partial_2\right\rangle-\left\langle\vec{\omega}, \nabla_{\partial_1} \partial_2\right\rangle, \quad \forall  \partial_1, \partial_2\in \mathbf{Vect}(\mathcal{M}).
\end{equation}
Note that this covariant derivative is well defined due to Lemma~\ref{pairtancotan}.

Recall the flat coordinates \eqref{flatt}--\eqref{flath} of $\mathcal{M}$. In the loop space of  $\mathcal{M}$ with a loop parameter $x$, let us consider the vector field
\[
\frac{\partial}{\partial x}=\sum_{u\in\mathbf{u}}\frac{\partial u}{\partial x}\frac{\partial}{\partial u}.
\]
Inspired by \eqref{ham1}, let us introduce the following
Hamiltonian operator:
\begin{equation}\label{hamdef}
	\mathcal{P}_1(\vec{\omega}) = \eta \left( \nabla_{\frac{\partial}{\partial x}} \vec{\omega} \right),
\end{equation}
with $\vec{\omega}$ being an arbitrary covector field.

\begin{prop}\label{P1}
	The Hamiltonian operator \eqref{hamdef} satisfy
		\begin{align}
			\mathcal{P}_1(\vec{\omega})=&\Big(\{\omega(z),a(z)\}_{-}+\{\hat{\omega}(z),\hat{a}(z)\}_{-} -\{\omega(z)_{-}+\hat{\omega}(z)_{-},a(z)\},
\nn\\
			&-\{\omega(z),a(z)\}_{+}-\{\hat{\omega}(z),\hat{a}(z)\}_{+}+\{\omega(z)_{+} +\hat{\omega}(z)_{+},\hat{a}(z)\}\Big), \label{whiham}
		\end{align}
	where $\vec{\omega}=(\omega(z),\hat{\omega}(z))$ is a covector field on $\mathcal{M}$, and
\[
\{f,g\}=\frac{\partial f}{\partial z}\frac{\partial g}{\partial x}-\frac{\partial g}{\partial z}\frac{\partial f}{\partial x}.
\]
\end{prop}
\begin{proof}
	Clearly, one has
	$$
	\partial_1\left\langle\vec{\omega}, \partial_2\right\rangle=\dsum_{i=1}^{m}\oint_{\gamma_{i}}\left(\partial_1 \omega(z) \partial_2 a(z)+\omega(z) \partial_1 \partial_2 a(z)+\partial_1 \hat{\omega}(z) \partial_2 \hat{a}(z)+\hat{\omega}(z) \partial_1 \partial_2 \hat{a}(z)\right) d z.
	$$
According to Lemma~\ref{nabla}, we have
		\begin{align*}
			& \left\langle\vec{\omega}, \nabla_{\partial_1} \partial_2\right\rangle\\
			= &\frac{1}{2\pi\mathrm{i}}\dsum_{s=1}^{m}\oint_{\gamma_{s}}\omega(z) \bigg(\partial_1 \partial_2 a(z)-\left(\frac{\partial_1 \zeta(z) \partial_2 \zeta(z)}{\zeta^{\prime}(z)}\right)_{-}^{\prime}-\left(\frac{\partial_1 \ell(z) \partial_2 \ell(z)}{\ell^{\prime}(z)}\right)_{\infty,\,\ge 0}^{\prime}
\\
&-\sum_{i=1}^{m}\left(\frac{\partial_1 \ell(z) \partial_2 \ell(z)}{\ell^{\prime}(z)}\right)_{\varphi_{i},\,\le -1}^{\prime}\bigg) d z \\
			& +\frac{1}{2\pi\mathrm{i}}\dsum_{s=1}^{m}\oint_{\gamma_{s}}\hat{\omega}(z)\bigg( \partial_1 \partial_2 \hat{a}(z)+\left(\frac{\partial_1 \zeta(z) \partial_2 \zeta(z)}{\zeta^{\prime}(z)}\right)_{-}^{\prime}-\left(\frac{\partial_1 \ell(z) \partial_2 \ell(z)}{\ell^{\prime}(z)}\right)_{\infty,\,\ge 0}^{\prime}
\\
&-\sum_{i=1}^{m}\left(\frac{\partial_1 \ell(z) \partial_2 \ell(z)}{\ell^{\prime}(z)}\right)_{\varphi_{i},\,\le -1}^{\prime}\bigg) d z \\
			= & \frac{1}{2\pi\mathrm{i}}\dsum_{s=1}^{m}\oint_{\gamma_{s}}\bigg(\omega(z) \partial_1 \partial_2 a(z)+\hat{\omega}(z) \partial_1 \partial_2 \hat{a}(z)+\omega^{\prime}(z)_{+}\left(\frac{\partial_1 \zeta(z) \partial_2 \zeta(z)}{\zeta^{\prime}(z)}\right)
\\
&-\hat{\omega}^{\prime}(z)_{-}\left(\frac{\partial_1 \zeta(z) \partial_2 \zeta(z)}{\zeta^{\prime}(z)}\right)\bigg) d z \\
			& -\operatorname{Res}_{z=\infty}\left(\omega^{\prime}(z)_{-}\left(\frac{\partial_1 \ell(z) \partial_2 \ell(z)}{a^{\prime}(z)}\right)+\hat{\omega}^{\prime}(z)_{-}\left(\frac{\partial_1 \ell(z) \partial_2 \ell(z)}{a^{\prime}(z)}\right)\right) d z \\
			& +\sum_{i=1}^{m}\operatorname{Res}_{z=\varphi_{i}}\left(\omega^{\prime}(z)_{+}\left(\frac{\partial_1 \ell(z) \partial_2 \ell(z)}{\hat{a}^{\prime}(z)}\right)+\hat{\omega}^{\prime}(z)_{+}\left(\frac{\partial_1 \ell(z) \partial_2 \ell(z)}{\hat{a}^{\prime}(z)}\right)\right) d z .
		\end{align*}
Denote $X=\eta\left(\nabla_{\p_{1}}\vec{\omega}\right)$, then by using \eqref{pairmec} and \eqref{nablaom} we have
	\begin{align*}
	\left\langle X,  \partial_2\right\rangle=&
	\partial_1\left\langle\vec{\omega}, \partial_2\right\rangle-\left\langle\vec{\omega}, \nabla_{\partial_1} \partial_2\right\rangle \\
		= &\frac{1}{2\pi\mathrm{i}}\dsum_{s=1}^{m}\oint_{\gamma_{s}}\left(\partial_1 \omega(z)-\omega^{\prime}(z)_{+} \frac{\partial_1 \zeta(z)}{\zeta^{\prime}(z)}+\hat{\omega}^{\prime}(z)_{-} \frac{\partial_1 \zeta(z)}{\zeta^{\prime}(z)}\right) \partial_2 a(z) dz\\
		&+\frac{1}{2\pi\mathrm{i}}\dsum_{s=1}^{m}\oint_{\gamma_{s}} \left(\partial_1 \hat{\omega}(z)+\omega^{\prime}(z)_{+} \frac{\partial_1 \zeta(z)}{\zeta^{\prime}(z)}-\hat{\omega}^{\prime}(z)_{-} \frac{\partial_1 \zeta(z)}{\zeta^{\prime}(z)}\right) \partial_2 \hat{a}(z) d z\\
		&+\operatorname{Res}_{z=\infty}\left(\left(\omega^{\prime}(z)_{-}+\hat{\omega}^{\prime}(z)_{-}\right) \frac{\partial_1 \ell(z)}{a^{\prime}(z)}\right) \partial_2 a(z) d z \\
		&-\sum_{i=1}^{m}\operatorname{Res}_{z=\varphi_{i}}\left(\left(\omega^{\prime}(z)_{+}+\hat{\omega}^{\prime}(z)_{+}\right) \frac{\partial_1 \ell(z)}{\hat{a}^{\prime}(z)}\right) \partial_2 \hat{a}(z) d z \\
		= & \frac{1}{2\pi\mathrm{i}}\dsum_{s=1}^{m}\oint_{\gamma_{s}}\frac{((\p_{1}\omega_{+}-\p_{1}\hat{\omega}_{-})\zeta'-(\omega_{+}'-\hat{\omega}_{-}')\p_{1}\zeta)\p_{2}\zeta}{\zeta'}dz\\
		&-\res_{z=\infty}\frac{((\p_{1}\omega_{-}+\p_{1}\hat{\omega}_{-})\ell'-(\omega_{-}'+\hat{\omega}_{-}')\p_{1}\ell)\p_{2}\ell}{\ell'}dz\\
		&-\dsum_{i=1}^{m}\res_{z=\varphi_{i}}\frac{((\omega_{-}'+\hat{\omega}_{-}')\p_{1}\ell-(\p_{1}\omega_{-}+ \p_{1}\hat{\omega}_{-})\ell')\p_{2}\ell}{\ell'}dz,
	\end{align*}
which implies that
\begin{align*}
\p_{X}\zeta(z)=&\p_{1}\zeta(z)(\omega'(z)_{+}-\hat{\omega}'(z)_{-})-\zeta'(z)(\p_{1}\omega(z)_{+}-\p_{1}\hat{\omega}(z)_{-}), \\
\p_{X}\ell(z)=&((\p_{1}\omega(z)_{-}+\p_{1}\hat{\omega}(z)_{-})a'(z)-(\omega'(z)_{-}+\hat{\omega}'(z)_{-})\p_{1}a(z))_{\infty,\,\ge 0}\\		&+\dsum_{i=1}^{m}((\omega'(z)_{-}+\hat{\omega}'(z)_{-})\p_{1}\hat{a}(z)-(\p_{1}\omega(z)_{-}+\p_{1}\hat{\omega}(z)_{-})\hat{a}'(z))_{\varphi_{i},\,\le -1}.
\end{align*}
In particular, let us take $\p_{1}=\p/\p{x}$, then we obtain
	\begin{align*}
		&\p_{X}\zeta(z)=\{\omega(z)_{+}-\hat{\omega}(z)_{-},\zeta(z)\},\\
		&\p_{X}\ell(z)=-\{\omega(z)_{-}+\hat{\omega}(z)_{-},a(z)\}_{+}+\{\omega(z)_{+}+\hat{\omega}(z)_{+},\hat{a}(z)\}_{-}.
	\end{align*}
	Hence
	\begin{align*}
		&	\p_{X}a(z)=\{\omega(z),a(z)\}_{-}+\{\hat{\omega}(z),\hat{a}(z)\}_{-}-\{\omega(z)_{-}+\hat{\omega}(z)_{-},a(z)\},\\
		&\p_{X}\hat{a}(z)=-\{\omega(z),a(z)\}_{+}-\{\hat{(z)\omega(z)},\hat{a}(z)\}_{+}+\{\omega(z)_{+}+\hat{\omega}(z)_{+},\hat{a}(z)\}.
	\end{align*}
	The proposition is proved.
\end{proof}

Let $\Sigma$ be the subset of $\mathcal{M}$ on which the map $g$ defined in  \eqref{gmap} is not invertible. We introduce
another Hamiltonian operator $\mathcal{P}_2$ from the space of covector fields to that of vector fields on $\mathcal{M} \setminus \Sigma$ as
\begin{equation}\label{}
\mathcal{P}_2(\vec{\omega})=g\left(\nabla_{\frac{\p}{\p x}}\vec{\omega}\right).
\end{equation}
Similar as above, after a lengthy but straightforward calculation, the following result is obtained.
\begin{prop}
For an arbitrary covector field $\vec{\omega}=(\omega(z),\hat{\omega}(z))$ on  $\mathcal{M} \setminus \Sigma$, the Hamiltonian operator ${\mathcal{P}_2}$ satisfies
	\begin{align}
	{\mathcal{P}_2}(\vec{\omega})=&\Big((\{\omega,a\}_{-}+\{\hat{\omega},\hat{a}\}_{-})a-\{(\omega a+\hat{\omega}\hat{a})_{-},a\}-\rho a',\nn\\
		&-(\{\omega,a\}_{+}+\{\hat{\omega},\hat{a}\}_{+})\hat{a}+\{(\omega a+\hat{\omega}\hat{a})_{+},\hat{a}\} -\rho \hat{a}' \Big), \label{intersec}
	\end{align}
where
\[
\rho= \frac{1}{2 n_{0}\pi\mathrm{i}}\sum_{i=1}^{m}\oint_{\gamma_{i}}(\{\omega(z),a(z)\}+\{\hat{\omega}(z),\hat{a}(z)\})dz.
\]
\end{prop}
%

\subsection{Densities of Hamiltonian functionals}
Now let us introduce a class of densities of Hamiltonian functionals, which are supposed to fulfill conditions of the form \eqref{princon4}--\eqref{princon2}.
More exactly, we introduce a class of analytic functions on $\mathcal{M}$ as:
\begin{align}
	\theta_{t_{i,s},p}=&\frac{1}{2 \pi \mathrm{i}} \frac{1}{(p+1) !} \frac{d_{i}}{s+d_{i}} \oint_{\gamma_{i}} \zeta(z)^{\frac{s}{d_{i}}} \left(a(z)^{p+1} -\hat{a}(z)^{p+1}\right) d z,\quad s\ne -d_{i};
\label{theta1}
\\
	\theta_{t_{i,-d_{i}},p}=&\frac{d_{i}}{2 \pi \mathrm{i}} \oint_{\gamma_{i}} \frac{a(z)^p}{p !}\log \frac{\zeta(z)^{\frac{1}{d_{i}}}}{(z-\varphi_{i})}dz+d_{i}\res_{z=\infty} \frac{a(z)^{p}}{p!}\left(\log \frac{a(z)^{\frac{1}{n_{0}}}}{z-\varphi_{i}}-\frac{c_{p}}{n_{0}}\right) dz
\nn\\
& -\frac{d_{i}}{2 \pi \mathrm{i}} \dsum_{i'\ne i}\oint_{\gamma_{i'}}\frac{a(z)^{p}}{p!}\log(z-\varphi_{i})dz;
\\
\theta_{h_{0,j},p}=& -\frac{\Gamma\left(1-\frac{j}{n_{0}}\right)}{\Gamma\left(2+p-\frac{j}{n_{0}}\right)} \res_{z=\infty} a(z)^{1+p-\frac{j}{n_{0}}} d z;
\\
\theta_{h_{k,r},p}=&\frac{\Gamma\left(1-\frac{r}{n_{k}}\right)}{\Gamma\left(2+p-\frac{r}{n_{k}}\right)} \res_{z=\varphi_{k}} \hat{a}(z)^{1+p-\frac{r}{n_{k}}} d z,\quad r\ne n_{k};
\\
\theta_{h_{k,n_{k}},p}=&\frac{n_{k}}{2\pi\mathrm{i}}\oint_{\gamma_{k}} \frac{\hat{a}(z)^{p}}{p!}\log\frac{\zeta(z)^{\frac{1}{d_{k}}}}{z-\varphi_{k}}dz
\nn\\
& +n_{k} \res_{z=\varphi_{k}}\frac{\hat{a}(z)^{p}}{p!}\left(\log \hat{a}(z)^{\frac{1}{n_{k}}}(z-\varphi_{k})-\frac{c_{p}}{n_{k}}\right)dz -\frac{n_{k}}{d_{k}}\theta_{t_{k,-d_{k}},p}, \label{theta5}
\end{align}
where
	$$
c_0=0, \quad c_p=\sum_{s=1}^p \frac{1}{s} .
	$$

\begin{thm}\label{mainthm2}
	The functions \eqref{theta1}--\eqref{theta5} satisfy an analogue of the system of equations  \eqref{princon4}--\eqref{princon2}. Namely, for any $u,v\in \mathbf{u}$ and $p\in\mathbb{Z}_{\ge1}$ it holds that:
\begin{align}\label{princon1u}
 \theta_{u,0}=&\sum_{v\in\mathbf{u}}\left\langle\frac{\p}{\p u},\frac{\p}{\p v}\right\rangle_\eta v, \\
\nabla_{\p /\p v}d \theta_{u,p}=&\eta^{-1}\left( C_{\p/\p v}(d \theta_{u,p-1}) \right),  \label{princon2u}
\\
	\operatorname{Lie}_E \left(\frac{\partial \theta_{u, p} }{\partial v}\right) =& \left(p + \mu_u + \mu_v\right) \frac{\partial \theta_{u, p} }{\partial v} +\sum_{w\in\mathbf{u}} \left(R\right)_u^w \frac{\partial \theta_{w, p-1}}{\partial v} . \label{princon3u}
\end{align}
where the spectrum $\mu$ and the nilpotent matrix \(R\) of \(\mathcal{M}\) are given by
	\begin{equation*}
		\mu_u = \begin{cases}
			\frac{s}{d_{i}} + \frac{1}{2}, &  u = t_{i,s}; \\
			\frac{1}{2} - \frac{j}{n_{0}}, &  u = h_{0,j}; \\
			\frac{1}{2} - \frac{r}{n_{k}}, &  u = h_{k,r},
		\end{cases}
	\end{equation*}
	\begin{equation*}
		(R)_{v}^{u}=\left(R_1\right)_v^u= \begin{cases}\delta_{i,j}-\frac{d_{i}}{n_{0}}, & u\in\cup_{j=1}^{m}\{t_{j,0},h_{j,0}\},\  v=t_{i,-d_{i}};\\
			\frac{n_{i}}{n_{0}}+1, & u=h_{i,0},\  v=h_{i,n_{i}};	 \\
			\frac{n_{i}}{n_{0}}-\frac{n_{i}}{d_{i}}, & u=t_{i,0},\  v=h_{i,n_{i}};	 \\
			\frac{n_{i}}{n_{0}}, & u\in\cup_{j\ne i}\{t_{j,0},h_{j,0}\},\  v=h_{i,n_{i}};	\\
			0, & \text { else. }\end{cases}
	\end{equation*}
\end{thm}

To prove this theorem, the following lemma is useful.
\begin{lem}
Assume \( Q_{p}(y, \hat{y})\) with $ p \in \mathbb{Z}_{\ge0} $ to be a series of smooth functions of \( y \) and \( \hat{y} \), satisfying
\[
\left(\frac{\partial }{\partial y} + \frac{\partial}{\partial \hat{y}}\right) Q_{p}(y, \hat{y}) = Q_{p-1}(y, \hat{y}).
\]
Let
\begin{equation}\label{Fdef}
	F_{i,p} = \frac{1}{2\pi\mathrm{i}} \oint_{\gamma_{i}} Q_{p+1}(a(z), \hat{a}(z)) dz,
\end{equation}
then for any vector field $\p$ on $\mathcal{M}$ it holds that
\begin{equation}\label{Fcond}
	\eta \left( \nabla_{\partial} d F_{i,p} \right) = C_{\partial}(d F_{i,p-1}), \quad p \in \mathbb{Z}_{\ge0}.
\end{equation}
\end{lem}
\begin{proof}
For convenience, let us use notations like
\[
\frac{\p Q_{p}}{\p a}=\frac{\p Q_{p}}{\p y}(a(z),\hat{a}(z)), \quad  \frac{\p Q_{p}}{\p \hat{a}}=\frac{\p Q_{p}}{\p \hat{y}}(a(z),\hat{a}(z)), \quad \{f,g\}_{\p}=f' \p g-g' \p f.
\]
It is easy to see
\begin{equation}\label{eq-Qaa}
	\left\{\frac{\p Q_{p}}{\p a},a\right\}_{\p}=\left\{\frac{\p Q_{p}}{\p \hat{a}},\hat{a}\right\}_{\p}=0.
\end{equation}
One has
	$$
	\left.dF_{i,p}\right|_{\vec{a}}=\left.\left(\frac{\p Q_{p+1}}{\p a}\mathbf{1}_{\gamma_{i}},\frac{\p Q_{p+1}}{\p \hat{a}}\mathbf{1}_{\gamma_{i}}\right)\right|_{\vec{a}}, \quad \vec{a}=(a(z),\hat{a}(z))\in\mathcal{M}.
	$$
Similar as in Proposition~\ref{P1} (with the bracket $\{\ ,\ \}$ replaced by  $\{\ ,\ \}_\p$), and by using the equalities \eqref{eq-Qaa}, we obtain
	$$
\eta\left(\nabla_{\p}dF_{i,p}\right)=\left(-\left\{(Q_{p}\mathbf{1}_{\gamma_{i}})_{-},a\right\}_{\p}, \left\{(Q_{p}\mathbf{1}_{\gamma_{i}})_{+},\hat{a}\right\}_{\p}\right).
	$$
On the other hand, thanks to \eqref{whiprod} we have
	\begin{align*}
		C_{\p}(dF_{i,p-1})=&\left( a'\left(\frac{\p Q_{p}}{\p a}\p a\mathbf{1}_{\gamma_{i}}+\frac{\p Q_{p}}{\p \hat{a}}\p \hat{a}\mathbf{1}_{\gamma_{i}}\right)_{-}-\p a\left(\frac{\p Q_{p}}{\p a}a'\mathbf{1}_{\gamma_{i}}+\frac{\p Q_{p}}{\p \hat{a}} \hat{a}'\mathbf{1}_{\gamma_{i}}\right)_{-} \right.,\\
		&\left. -\hat{a}'\left(\frac{\p Q_{p}}{\p a}\p a\mathbf{1}_{\gamma_{i}}+\frac{\p Q_{p}}{\p \hat{a}}\p \hat{a}\mathbf{1}_{\gamma_{i}}\right)_{+}+\p \hat{a}\left(\frac{\p Q_{p}}{\p a}a'\mathbf{1}_{\gamma_{i}}+\frac{\p Q_{p}}{\p \hat{a}} \hat{a}'\mathbf{1}_{\gamma_{i}}\right)_{+}\right)\\
		=&\left(-\{\left(Q_{p}\mathbf{1}_{\gamma_{i}}\right)_{-},a\}_{\p}, \{\left(Q_{p}\mathbf{1}_{\gamma_{i}}\right)_{+},\hat{a}\}_{\p}\right).
	\end{align*}
Thus the equality \eqref{Fcond} is verified. Therefore, the lemma is proved.
\end{proof}

\begin{proof}[Proof of Theorem \ref{mainthm2}]
Firstly, the equalities \eqref{princon1u} are verified with the help of Proposition~\ref{etaflat}.

Secondly, the Hamiltonian density \( \theta_{t_{i,s},p} \), with \( s \neq -d_{i} \), takes the form given in \eqref{Fdef} and therefore satisfies equation \eqref{Fcond}, which is equivalent to \eqref{princon2u}. For \( \theta_{t_{i,-d_{i}},p} \), consider a subset \( \mathcal{M}' \) of \( \mathcal{M} \) such that \( a^{\frac{1}{n_{0}}}(z) \) is analytically continued to \( \cup_{s=1}^{m} \gamma_{s} \), with winding number 1 around \( \gamma_{i} \) but 0 around \( \gamma_{s} \) (\( s \neq i \)). Consequently, on the subset $\mathcal{M}'$ one has
	\begin{align*}
\theta_{t_{i,-d_{i}},p}=&-\frac{d_{i}c_{p}}{n_{0}}\res_{z=\infty}\frac{a^{p}}{p!}dz-\frac{d_{i}}{2\pi\mathrm{i}} \oint_{\gamma_{i}}\frac{a^{p}}{p!}\log \frac{a^{ {1}/{n_{0}}}}{ \zeta^{{1}/{d_{i}}} } dz-\frac{d_{i}}{2\pi\mathrm{i}}\dsum_{s\ne i}\oint_{\gamma_{s}}\frac{a^{p}}{p!}\log a^{\frac{1}{n_{0}}}dz\\
	=&-\frac{d_{i}}{2\pi\mathrm{i}}\oint_{\gamma_{i}}\frac{a^{p}}{p!}\left(\log \frac{a^{ {1}/{n_{0}}}}{ \zeta^{{1}/{d_{i}}} } -\frac{c_{p}}{n_{0}}\right)dz -\frac{d_{i}}{2\pi\mathrm{i}}\dsum_{s\ne i}\oint_{\gamma_{s}}\frac{a^{p}}{p!}\left(\log a^{\frac{1}{n_{0}}}-\frac{c_{p}}{n_{0}}\right)dz,
\end{align*}
which is a linear combination of functions of the form \eqref{Fdef}, hence it satisfies \eqref{princon1} and \eqref{princon2u}. Moreover, since \( \theta_{t_{i,-d_{i}},p} \) is analytic on $\mathcal{M}$, then the equality  \eqref{princon2u} in fact holds on $\mathcal{M}$ by property of analytic functions. The other cases of \eqref{princon2u} can be verified with a similar method.
%

In order to show \eqref{princon3u}, based on $E$ given in \eqref{E}, let us introduce
\[
\mathcal{E} = E + \frac{1}{n_{0}} z \frac{\partial}{\partial z}.
\]
It is easy to see:
\[
\operatorname{Lie}_{\mathcal{E}} \zeta(z) = \zeta(z), \quad \operatorname{Lie}_{\mathcal{E}} a(z) = a(z), \quad \operatorname{Lie}_{\mathcal{E}} \hat{a}(z) = \hat{a}(z), \label{homo}
\]
and that, for any analytic function \( f(y, \hat{y}) \),
$$
\operatorname{Lie}_E \oint_{\gamma_{s}} f(a(z), \hat{a}(z)) d z=\oint_{\gamma_{s}} \operatorname{Lie}_{\mathcal{E}} f(a(z), \hat{a}(z)) d z + \frac{1}{n_{0}} \oint_{\gamma_{s}} f(a(z), \hat{a}(z)) d z,
$$
where $ s=1,\cdots,m$.
Using these equalities, it is straightforward to compute:
{\small
\begin{align*}
	\operatorname{Lie}_{E}\theta_{t_{i,s},p}=&\left(p+1+\frac{s}{d_{i}}+\frac{1}{n_{0}}\right)\theta_{t_{i,s},p},\quad s\ne -d_{i};
\\
\operatorname{Lie}_{E}\theta_{t_{i,-d_{i}},p}=&\left(p+\frac{1}{n_{0}}\right)\theta_{t_{i,-d_{i}},p} +\left(1-\frac{d_{i}}{n_{0}}\right)\left(\theta_{t_{i,0},p-1}+\theta_{h_{i,0},p-1}\right)
\\
&-\dsum_{i'\ne i}\frac{d_{i}}{n_{0}}\left(\theta_{t_{i',0},p-1}+\theta_{h_{i',0},p-1}\right);
\\
	\operatorname{Lie}_{E}\theta_{h_{0,j},p}=&\left(p+1-\frac{j}{n_{0}}+\frac{1}{n_{0}}\right)\theta_{h_{0,j},p};
\\
	\operatorname{Lie}_{E}\theta_{h_{k,r},p}=&\left(p+1-\frac{r}{n_{k}}+\frac{1}{n_{0}}\right)\theta_{h_{k,r},p},\quad r\ne n_{k};
\\
\operatorname{Lie}_{E}\theta_{h_{k,n_{k}},p}=&\left(p+\frac{1}{n_{0}}\right)\theta_{h_{k,n_{k}},p} +\left(\frac{n_{k}}{n_{0}}+1\right)\theta_{h_{k,0},p-1} +\left(\frac{n_{k}}{n_{0}}-\frac{n_{k}}{d_{k}}\right)\theta_{t_{k,0},p-1} \\
&+\dsum_{k'\ne k}\frac{n_{k}}{n_{0}}\left(\theta_{t_{k',0},p-1}+\theta_{h_{k',0},p-1}\right).
\end{align*}}
Thus the equalities \eqref{princon3u} are confirmed.
The theorem is proved.
\end{proof}

\subsection{Principal hierarchy and the genus-zero Whitham hierarchy}

In combination of the data of the previous two subsections, now we are ready to compute the principal hierarchy for $\mathcal{M}$:
 \[
 \frac{\partial\vec{a}}{\partial T^{u,p}} = \mathcal{P}_1 (d\theta_{u,p+1}), \quad u\in\mathbf{u}, ~~ p\ge0,
 \]
and then study its relationship with the genus-zero Whitham hierarchy. In other words, let us proceed to prove Theorem~\ref{mainthm3} and Corollary~\ref{prinwhi}.

 \begin{proof}[Proof of Theorem \ref{mainthm3}]
According to the definition of the Hamiltonian densities \( \theta_{h_{0,j},p+1} \), one has, for $\vec{a}=(a,\hat{a})\in\mathcal{M}$, 
 	$$
 	d\theta_{h_{0,j},p+1}=(K_{h_{0,j},p},0)|_{\vec{a}}, \quad K_{h_{0,j},p}=\frac{\Gamma\left(1-\frac{j}{n_{0}}\right)}{\Gamma\left(2+p-\frac{j}{n_{0}}\right)} a^{1+p-\frac{j}{n_{0}}}.
 	$$
Hence, by using \eqref{whiham} we obtain
 	\begin{align*}
 		\mathcal{P}_1( d\theta_{h_{0,j},p+1})=&\left(\{(K_{h_{0,j},p})_{+},a\},\{(K_{h_{0,j},p})_{+},\hat{a}\}\right)\\
 =& \frac{\Gamma\left(1-\frac{j}{n_{0}}\right)}{\Gamma\left(2+p-\frac{j}{n_{0}}\right)} \left(\{(a^{1+p-\frac{j}{n_{0}}})_{\infty,\,\ge 0},a\},\{(a^{1+p-\frac{j}{n_{0}}})_{\infty,\,\ge 0},\hat{a}\}\right).
 	\end{align*}
%
%
With the same method, for the Hamiltonian density \( \theta_{h_{k,r},p+1} \) with \( r \neq n_{k} \), one has
 	\[
 	d\theta_{h_{k,r},p+1} = (0, K_{h_{k,r},p}\mathbf{1}_{\gamma_{k}})|_{\vec{a}}, \quad K_{h_{k,r},p} = \frac{\Gamma\left(1 - \frac{r}{n_{k}}\right)}{\Gamma\left(2 + p - \frac{r}{n_{k}}\right)} \hat{a}^{1 + p - \frac{r}{n_{k}}}.
 	\]
Hence
 	\begin{align*}
 		\frac{\Gamma\left(2 + p - \frac{r}{n_{k}}\right)}{\Gamma\left(1 - \frac{r}{n_{k}}\right)} \mathcal{P}_1(d\theta_{h_{k,r},p+1}) =& \left( \{-( \hat{a}^{1 + p - \frac{r}{n_{k}}} \mathbf{1}_{\gamma_{k}})_{-}, a\}, \{-( \hat{a}^{1 + p - \frac{r}{n_{k}}}\mathbf{1}_{\gamma_{k}})_{-}, \hat{a}\} \right) \\
 		=& \left( \{-( \hat{a}^{1 + p - \frac{r}{n_{k}}})_{\varphi_{k}, \leq -1}, a\}, \{-( \hat{a}^{1 + p - \frac{r}{n_{k}}})_{\varphi_{k}, \leq -1}, \hat{a}\} \right).
 	\end{align*}
 The other cases are similar. Thus the theorem is proved.
\end{proof}

\begin{proof}[Proof of Corollary~\ref{prinwhi}]

Note that under the assumption \eqref{assum}, the genus-zero Whitham hierarchy  \eqref{whi1}-\eqref{whi3} is recast to:
 \begin{align*}
 	&\frac{\partial a(z)}{\partial s^{0,p}} = \left\{\left(a(z)^{\frac{p}{n_{0}}}\right)_{\infty,\,\ge 0}, a(z)\right\}, \quad \frac{\partial \hat{a}(z)}{\partial s^{0,p}} = \left\{\left(a(z)^{\frac{p}{n_{0}}}\right)_{\infty,\,\ge 0}, \hat{a}(z)\right\}, \\
 	&\frac{\partial a(z)}{\partial s^{j,p}} = \left\{-\left(\hat{a}(z)^{\frac{p}{n_{j}}}\right)_{\varphi_{j},\,\le -1}, a(z)\right\}, \quad \frac{\partial \hat{a}(z)}{\partial s^{j,p}} = \left\{-\left(\hat{a}(z)^{\frac{p}{n_{j}}}\right)_{\varphi_{j},\,\le -1}, \hat{a}(z)\right\}, \\
 	&\frac{\partial a(z)}{\partial s^{j,0}} = \left\{\log(z - \varphi_{j}), a(z)\right\}, \quad \frac{\partial \hat{a}(z)}{\partial s^{j,0}} = \left\{\log(z - \varphi_{j}), \hat{a}(z)\right\},
 \end{align*}
 where  $j\in\{1,\cdots,m\}$ and $p\in\mathbb{Z}_{>0}$. Comparing these evolutionary equations with those in Theorem~\ref{mainthm3}, we obtain
\begin{align*}
	&\frac{\p  }{\p s^{0,(p+1)n_{0}-j}}=\frac{\Gamma(2+p-\frac{j}{n_{0}})}{\Gamma(1-\frac{j}{n_{0}})}\frac{\p}{\p T^{h_{0,j},p}}, \quad j=1,2,\dots,n_0-1; \\
	&\frac{\p}{\p s^{i,(p+1)n_{i}-r}}=\frac{\Gamma(2+p-\frac{r}{n_{i}})} {\Gamma(1-\frac{r}{n_{i}})}\frac{\p}{\p T^{h_{i,r},p}}, \quad r=0,1,\dots,n_i-1; \\
	&\frac{\p}{\p s^{0,(p+1)n_{0}}}=(p+1)!\dsum_{k=1}^{m}\left(\frac{\p}{\p T^{t_{k,0},p}}+\frac{\p}{\p T^{h_{k,0},p}}\right),
\\
	&\frac{\p}{\p s^{i,0}}=\frac{1}{n_{i}}\frac{\p}{\p T^{h_{i,0},0}},
\end{align*}
where $i\in\{1,2,\dots,m\}$ and $p\in\mathbb{Z}_{\ge0}$. Therefore the corollary is proved.
 \end{proof}

\subsection{Reduction to finite-dimensional cases} \label{McKP}

Recall $\zeta(z)$ and $\ell(z)$ given in \eqref{zetaell}. Now suppose that $\zeta(z)\to 0$, and that
\begin{equation}\label{ell}
\ell(z)= 
z^{n_0} + a_{n_0-2}z^{n_0-2} + \cdots + a_1 z + a_{0} + \sum_{i=1}^{m}\sum_{j=1}^{n_i}b_{i,j}(z-\varphi_{i})^{-j}.
\end{equation}
Rational functions of this form compose a space denoted by \(  M^{cKP} \), for which a system of coordinates reads
\[
\left\{ a_k,\, b_{i,j_i}, \, \varphi_i \mid 0\le k\le n_0-2, \, 1\le i\le m, \, 1\le j_i\le n_i\right\},
\]
and another system of coordinates is given by $\mathbf{h}$ in \eqref{flath}.

Letting $\zeta(z)\to 0$, the infinite-dimensional Frobenius manifold $\mathcal{M}$ is reduced to a Frobenius manifold \(  M^{cKP} \) of finite dimension. More exactly, on \(  M^{cKP} \)  the flat metric reduced from \eqref{whimetric} is
 \[
\langle \p', \p'' \rangle_{\eta} = \sum_{|\lambda|<\infty} \res_{d\lambda=0} \frac{\p'(\lambda(z)dz) \p''(\lambda(z)dz)}{d\lambda(z)},
\]
and the symmetric $3$-tensor induced by that in Proposition~\ref{3tensorflat} reads
\[
c(\p', \p'', \p''') := \sum_{|\lambda|<\infty} \res_{d\lambda=0} \frac{\p'(\lambda(z)dz) \p''(\lambda(z)dz) \p'''(\lambda(z)dz)}{d\lambda(z)dz}.
\]
where \( \p', \p'', \p''' \in  \operatorname{Vect}(M^{cKP}) \).
Accordingly, the product between two vector field can be defined by
\[
\eta(\p' \circ \p'', \p''')=c(\p', \p'', \p''').
\]
Moreover, the unit vector field and the Euler vector field on $M^{cKP}$ are given by \eqref{whiid} and by \eqref{E}, respectively, with the $\p/\p t_{i,s}$-terms omitted.

Note that the Frobenius manifold $M^{cKP}$ appears   in \cite{dub1998} under a general setting for Frobenius manifolds on Hurwitz spaces. In particular, for the case $m=1$, a more concrete construction was given by Aoyama and Kodama \cite{aoyama1996topological}, who also clarified its relationship with the dispersionless limit of the constrained KP (cKP) hierarchy.

For the Frobenius manifold \( M^{cKP} \) with $m\ge2$, its principal hierarchy takes the form \cite{dub1998,ma2023principal2} (up to scaling of the flows):
{\small	\begin{align*}
	\frac{\partial \ell(z)}{\partial \tilde{T}^{{0,j};p-1}} =& \{(c_{0,j;p-1} w_{0}(z)^{p n_{0} - j})_{\infty, \geq 0}, \ell(z)\}, \quad j \in\{ 1, \ldots, n_{0} - 1\};
\\
	\frac{\partial \ell(z)}{\partial \tilde{T}{{i,j};p-1}} =& -\{(c_{i,j;p-1} w_{i}(z)^{p n_{i} - j})_{\varphi_i, \leq -1}, \ell(z)\}, \quad i \in\{ 1, \ldots, m\}, \ j \in\{ 0, \ldots, n_{i} - 1\};
\\
	\frac{\partial \ell(z)}{\partial \tilde{T}^{{i,n_{i}};p-1}} =& \bigg\{\left(\frac{\ell(z)^{p}}{p!} (\log \frac{w_{0}(z)}{z - \varphi_i} - \frac{c_{p}}{n_{0}})\right)_{\infty, \geq 0}
-\left(\frac{\ell(z)^{p}}{p!} (\log(z - \varphi_i) w_{i}(z) - \frac{c_{p}}{n_{i}})\right)_{\varphi_i, \leq -1}, \ell(z) \bigg\}\\
	& - \sum_{s \neq i} \left\{\left(\frac{\ell(z)^{p}}{p!} \log(z - \varphi_i)\right)_{a_{s,0}, \leq -1}, \ell(z)\right\} + \left\{\frac{\ell(z)^{p}}{p!} \log(z - \varphi_i), \ell(z)\right\},
\end{align*}
}
	where
\[
w_{i}(z) =
\begin{cases}
	(\ell(z))^{\frac{1}{n_0}} = z + w_{0,1}z^{-1} + \cdots ~~ \text{as} ~~ z \to \infty,\quad &i = 0; \\
	(\ell(z))^{\frac{1}{n_i}} = w_{i,1}(z-\varphi_i)^{-1} + \cdots ~~ \text{as} ~~ z \to \varphi_i,\quad &i = 1, \cdots, m,
\end{cases}
\]
with constants 
\[
c_{i,j;p} = \frac{n_{i}}{n_{i} - j} \frac{n_{i}}{2n_{i} - j} \cdots \frac{n_{i}}{(p + 1)n_{i} - j}, \quad c_{p} = \sum_{s = 1}^{p} \frac{1}{s}.\]
One observes that the flows $\frac{\p}{\p \tilde{T}^{{i,j;p}}}$ are just reductions of  $\frac{\p}{\p T^{h_{i,j;p}}}$ in the principal hierarchy for $\mathcal{\mathcal{M}}$.


\section{Conclusion}

In this paper, we have constructed a class of infinite-dimensional Frobenius manifolds on the spaces of pairs of meromorphic functions defined on certain regions of the Riemann sphere. This is a rather nontrivial generalization of the previous results obtained in \cite{ma2021infinite}, from the case of two marked points to that of $m+1$ marked points on the Riemann sphere. Moreover, for each of the infinite-dimensional Frobenius manifolds, we have derived its principal hierarchy, which is shown to be an extension of the genus-zero Whitham hierarchy. We hope that these results will be helpful not only to understand infinite-dimensional Frobenius manifolds, but also to study the dispersive Whitham hierarchies.
%
%

\medskip
{\small
\noindent{\bf Acknowledgements}.
The author C.-Z.Wu is partially supported by National Key R\&D Program of China 2023YFA1009801 and NSFC (No.12471243), and the author D.Zuo is
 partially supported by NSFC (No.12471241).
}
\bibliographystyle{unsrt}

\end{document}